\documentclass[12pt]{amsart}

\usepackage{amsmath, amsthm, amssymb,amscd}
\usepackage{fullpage, paralist,enumerate}
\usepackage{verbatim}
\usepackage[all]{xy}
\usepackage[parfill]{parskip}
\usepackage{graphicx}

\usepackage{afterpage}

\usepackage{amsmath}
\usepackage{amsfonts}
\usepackage{mathcomp}
\usepackage{textcomp}
\usepackage{amssymb}
\usepackage{latexsym}
\usepackage{graphicx}
\usepackage{courier}
\usepackage[english]{babel}
\usepackage{mathbbol}
\usepackage{multirow}
\usepackage{latexsym}
\usepackage{epsfig}
\usepackage{subfigure}
\usepackage{dcolumn}% Align table columns on decimal point
\usepackage{bm}% bold math
\usepackage{colordvi}
\usepackage{color}
\usepackage{booktabs}
\usepackage{stmaryrd}
\usepackage{cite}
\usepackage[linesnumbered,commentsnumbered,ruled]{algorithm2e}
\usepackage[noend]{algpseudocode}
\usepackage{epstopdf}
\input xy
\xyoption{all}
\newcommand{\commentout}[1]{}

\newtheorem{thm}{Theorem}[section]
\newtheorem{theorem}[thm]{Theorem}
\newtheorem{lem}[thm]{Lemma}
\newtheorem{prop}[thm]{Proposition}
\newtheorem{cor}[thm]{Corollary}
\newtheorem{ex}[thm]{Example}
\newtheorem{rmk}[thm]{Remark}
\newtheorem{lemma}[thm]{Lemma}

\newtheorem{corollary}[thm]{Corollary}

\newcommand{\nwc}{\newcommand*}

\nwc{\ben}{\begin{equation*}}
\nwc{\bea}{\begin{eqnarray}}
\nwc{\beq}{\begin{eqnarray}}
\nwc{\bean}{\begin{eqnarray*}}
\nwc{\beqn}{\begin{eqnarray*}}
\nwc{\beqast}{\begin{eqnarray*}}

\nwc{\eal}{\end{align}}
\nwc{\een}{\end{equation*}}
\nwc{\eea}{\end{eqnarray}}
\nwc{\eeq}{\end{eqnarray}}
\nwc{\eean}{\end{eqnarray*}}
\nwc{\eeqn}{\end{eqnarray*}}

\CompileMatrices

\theoremstyle{remark}

\nwc{\nn}{\nonumber}
\nwc{\mb}{\mathbf}
\nwc{\ml}{\mathcal}

\newcommand{\lt}{\left}
\newcommand{\rt}{\right}

\nwc{\vep}{\varepsilon}
\nwc{\ep}{\epsilon}
\nwc{\vrho}{\varrho}
\nwc{\orho}{\bar\varrho}
\nwc{\vpsi}{\varpsi}
\nwc{\lamb}{\lambda}
\nwc{\om}{\omega}
\nwc{\Om}{\Omega}
\nwc{\al}{\alpha}
\nwc{\sgn}{\mbox{\rm sgn}}

\nwc{\IA}{\mathbb{A}} %algebraic
\nwc{\bi}{\mathbf{i}}
\nwc{\ba}{\mathbf{a}}
\nwc{\bmb}{\mathbf{b}}
\nwc{\bo}{\mathbf{o}}
\nwc{\IB}{\mathbb{B}}
\nwc{\IC}{\mathbb{C}} %complex
\nwc{\ID}{\mathbb{D}} %Dedekind
\nwc{\IM}{\mathbb{M}} %Dedekind
\nwc{\IP}{\mathbb{P}} %Dedekind
\nwc{\II}{\mathbb{I}} %Dedekind
\nwc{\IE}{\mathbb{E}} %Euklides
\nwc{\IF}{\mathbb{F}} %finite field
\nwc{\IG}{\mathbb{G}} %Gauss
\nwc{\IN}{\mathbb{N}} %natural
\nwc{\IQ}{\mathbb{Q}} %rational
\nwc{\IR}{\mathbb{R}} %real
\nwc{\IT}{\mathbb{T}} %torus
\nwc{\IZ}{\mathbb{Z}} %integers

\nwc{\cE}{{\ml E}}
\nwc{\cP}{{\ml P}}
\nwc{\cQ}{{\ml Q}}
\nwc{\cL}{{\ml L}}
\nwc{\cX}{{\ml X}}
\nwc{\cW}{{\ml W}}
\nwc{\cZ}{{\ml Z}}
\nwc{\cR}{{\ml R}}
\nwc{\cV}{{\ml V}}
\nwc{\cT}{{\ml T}}
\nwc{\crV}{{\ml L}_{(\delta,\rho)}}
\nwc{\cC}{{\ml C}}
\nwc{\cO}{{\ml O}}
\nwc{\cA}{{\ml A}}
\nwc{\cK}{{\ml K}}
\nwc{\cB}{{\ml B}}
\nwc{\cD}{{\ml D}}
\nwc{\cF}{{\ml F}}
\nwc{\cS}{{\ml S}}
\nwc{\cM}{{\ml M}}
\nwc{\cG}{{\ml G}}
\nwc{\cH}{{\ml H}}
\nwc{\bk}{{\mb k}}
\nwc{\bn}{{\mb n}}
\nwc{\bp}{{\mb p}}
\nwc{\bz}{\mb z}
\nwc{\bl}{{\mb l}}
\nwc{\bj}{{\mb j}}
\nwc{\bs}{{\mb s}}
\nwc{\by}{\mathbf{h}}
\nwc{\bZ}{\mathbf{Z}}
\nwc{\bF}{\mathbf{F}}
\nwc{\bE}{\mathbf{E}}
\nwc{\bV}{\mathbf{V}}
\nwc{\bY}{\mathbf Y}
\nwc{\br}{\mb r}
\nwc{\pft}{\cF^{-1}_2}
\nwc{\bU}{{\mb U}}
\nwc{\bG}{{\mb G}}
\nwc{\bg}{\mathbf{g}}
\nwc{\mbf}{\mathbf{f}}
\nwc{\mbe}{\mathbf{e}}
\nwc{\be}{\mathbf{e}}
\nwc{\ind}{\operatorname{I}}
\nwc{\mbx}{\mathbf{f}}
\nwc{\bb}{\mathbf{g}}
\nwc{\xmax}{f_{\rm max}}
\nwc{\xmin}{f_{\rm min}}
\nwc{\suppx}{\hbox{\rm supp} (\mbf)}
\nwc{\cI}{\IZ^2_N}
\nwc{\chis}{{\chi^{\rm s}}}
\nwc{\chii}{{\chi^{\rm i}}}
\nwc{\pdfi}{{f^{\rm i}}}
\nwc{\pdfs}{{f^{\rm s}}}
\nwc{\pdfii}{{f_1^{\rm i}}}
\nwc{\pdfsi}{{f_1^{\rm s}}}
\nwc{\thetatil}{{\tilde\theta}}
\nwc{\red}{\color{red}}
\nwc{\blue}{\color{blue}}

\nwc{\prox}{\hbox{prox}}
\nwc{\diag}{\hbox{\rm diag}}
\nwc{\supp}{{\hbox{\rm supp}}}

\nwc{\sloc}{J_{\rm f}}
\nwc{\bu}{{\mb u}}
\nwc{\bv}{{\mb v}}
\nwc{\cU}{\mathcal{U}}
\nwc{\cN}{\mathcal{N}}
\nwc{\bN}{\mathbf{N}}
\nwc{\mbm}{\mathbf{m}}
\nwc{\bw}{\mathbf{w}}
\nwc{\bom}{\mathbf{w}}
\nwc{\bt}{\mathbf{t}}
\nwc{\z}{y}
\nwc{\cY}{\mathcal{Y}}
\nwc{\bM}{\mathbf{M}}
\nwc{\half}{{1\over 2}}
\nwc{\Sf}{S_{\rm f}}
\nwc{\Jf}{J_{\rm f}}
\nwc{\nul}{\hbox{\rm null}_\IR}
\nwc{\spanR}{\hbox{\rm span}_\IR}
\nwc{\Arg}{\hbox{\rm Arg~}}
\nwc{\fdr}{S_{\rm f}}
\nwc{\phase}[1]{\exp\lt[i\measured #1\rt]}

\nwc{\im}{{\rm i}}

\nwc{\cle}{\preccurlyeq}

\nwc{\lb}{\llbracket}
\nwc{\rb}{\rrbracket}
\nwc{\modpi}{{{\rm mod}\,2\pi}}
\nwc{\tphi}{{{\phi}_0}}
\nwc{\mpc}{\,\mbox{MPC($\gamma$)}\,\,}

\begin{document}

\centerline{\blue {\em Inverse Problems} {\bf 36} (2020) 045005 }
 \title{
Blind Ptychography: Uniqueness \& Ambiguities 
}

\author{Albert Fannjiang 
 \address{
Department of Mathematics, University of California, Davis, California  95616, USA. Email:  {\tt fannjiang@math.ucdavis.edu}
} \and Pengwen Chen
\address{ Applied Mathematics, National Chung Hsing University, Taichung 402, Taiwan. Email: {\tt pengwen@nchu.edu.tw}}
}

\maketitle 

\begin{abstract} Ptychography with an unknown mask and object is analyzed for general ptychographic measurement schemes
that are strongly connected and possess an anchor. 

Under a mild constraint on the mask phase,
it is proved that the masked object estimate must be the product of a block phase factor
and the true masked object.  This  local uniqueness manifests itself in the 
phase drift equation that determines the ambiguity at different locations connected by ptychographic shifts. 

The proposed mixing schemes effectively connects the ambiguity throughout the whole domain such that a distinct ambiguity
profile arises and consequently possess 
the  global uniqueness  that the block phases have an affine profile and
that the object and mask can be simultaneously recovered up to a constant scaling factor and
an affine phase factor.

\end{abstract}

%\begin{keywords} Ptychography, phase retrieval, affine phase ambiguity, uniqueness.
%\end{keywords}
% \begin{AMS}49K35, 05C70,  90C08\end{AMS}

\section{Introduction}

  \begin{figure}[t]
\begin{center}
%\subfigure[TCB]{
\includegraphics[width=10cm]{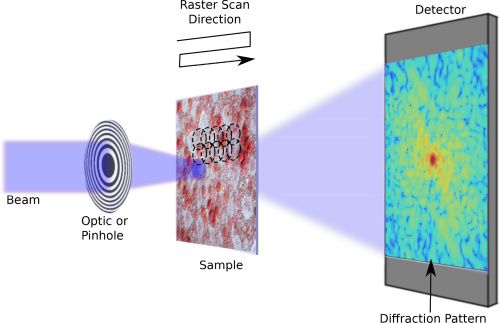}
%}\hspace{-0.6cm}
\caption{Simplified ptychographic setup showing a Cartesian grid used for the overlapping raster scan positions.
Adapted with permission from  \cite{parallel} \copyright The Optical Society.
}
\label{fig0}
\end{center}
\end{figure}

Ptychography is the scanning version of coherent diffractive imaging (CDI) \cite{CDI} that 
acquires  multiple diffraction patterns through the scan of a localized illumination on an extended object (Fig. \ref{fig0}). The redundant information in the overlap between adjacent illuminated spots is then exploited  to improve phase retrieval methods  \cite{Pfeiffer,Nugent,Rod08}. Ptychography originated in electron microscopy \cite{Hoppe1, Hoppe2, EM0, EM1,EM2,EM-ptych1,EM-ptych2} and has been successfully  implemented with X-ray, optical  and terahertz waves  \cite{ PIE104,PIE204,DM08, probe09, ptycho10,terahz,FROG}.   
Recently ptychography  has been extended to the Fourier domain \cite{FPM13, Yang14, add1}.  In Fourier ptychography,  illumination angles are scanned sequentially with a programmable array source with the diffraction pattern measured at each angle.

Ptychographic CDI has its origin in a concept developed for the crystallographic phase problem: Hoppe \cite{Hoppe1} pointed out that if one can make the Bragg peaks of crystalline diffraction patterns interfere, information about their relative phases can be obtained and therefore suggested to use a localized illumination instead of the usual extended plane wave. 
Due to the Fourier convolution theorem, the crystal's diffraction peaks in the resulting far-field pattern are then convolved with the Fourier transform of the localized illumination. When the extent of the illumination is shrunk to about the same order of magnitude as the crystalline unit cell, this leads to overlap between adjacent Bragg peaks and thus the desired interferences. While these interferences already allow to determine the relative phases, the twin-image ambiguity remains. Hoppe showed that an unambiguous result can be obtained by recording another diffraction pattern at a slightly shifted position of the localized illumination.  Hoppe \cite{Hoppe2} further discussed the extension of ptychography to non-periodic objects and
 the possibility of scanning transmission electron diffraction microscopy. 
 
  An important development in ptychography since the work of Thibault {\em et al.}  \cite{DM08,probe09} is the potential of simultaneous 
recovery of the object and the illumination (blind ptychography).
Blind ptychographic reconstruction  is affected by many factors such as
the type of illumination and  the amount of overlap between adjacent illuminations. In practice, numerical reconstruction with the widely used
algorithm, the extended Ptychographic Iterative Engine (ePIE), and its variants typically require 
60-70\% overlap between adjacent illuminations \cite{overlap,ePIE09,rPIE17} (see Section \ref{sec:con} for more discussion). 
The convergence of numerical reconstruction  is monitored with the residual of the ptychographic data or the difference between successive estimates  \cite{ePIE09, DM08,Fie08,ML12,Waller15, Hesse}. 

Even in the noiseless case, however, numerical convergence does not necessarily 
imply recovery of the mask and the object.  To ensure that a vanishing residual (data fitting) implies a vanishing reconstruction error in the noiseless case, we need a theory of uniqueness of solution. To be sure, a completely blind ptychography or phase retrieval is untenable. 

First of all,  even with a complete prior information of the mask/illumination, we have  shown in a recent work \cite{ptych-unique}   that twin-image ambiguity 
does arise if  the Fresnel number of the commonly used Fresnel illumination takes on certain values, resulting in poor
reconstruction and hinting on
the benefits of avoiding symmetry and increasing complexity of the mask.  
A simple way to avoid symmetry and increase complexity is to use a random mask for illumination.  Random masking is  a form of coded aperture and  has found applications in many imaging modalities and significant improvements
on imaging qualities  \cite{AH, Alm, BWW, rand, rpi, meta1, meta3,ptycho-rpi, random-aperture, diffuser, RCM, random-coding, meta2, Spread, Fucai1,Fucai2,Horisaki1,Horisaki2, rPIE17, optimal, ptych-unique,March16}.

For nonptychographic phase retrieval, the capability of a randomly coded aperture in removing all the ambiguities, including the translation and twin-image ambiguities, was rigorously analyzed in \cite{unique}.  
Moreover,  uniqueness theory for blind phase retrieval with a plain and a randomly coded diffraction pattern has been
developed in \cite{pum} which  assumes slight prior knowledge
about the phase range of the random mask. In other words, with  a plain and a randomly coded
diffraction pattern one can uniquely and simultaneously determine both the unknown object and the unknown mask. 
In contrast, in blind ptychography we work with just one unknown mask which is more challenging. 
As random masks are typically harder to calibrate (but easier to fabricate) than a deterministic mask, blind ptychography and phase retrieval is particularly useful when a random mask is used. 

This paper concerns the uniqueness question for blind ptychography with a randomly phased mask under certain prior information. 
We exhibit examples to show these priors are in some sense necessary. 
Moreover, we aim to characterize a general class of measurement schemes that avoid the pitfalls of the regular raster scan shown in Figure \ref{fig0} (see Examples \ref{ex6} and \ref{ex5}).

\subsection{Inherent ambiguities}

Let us begin with  two inherent ambiguities to blind ptychography.
 
% \subsection{Inherent ambiguities}

Let $\lb k,l\rb$ denote the integers
between and including the integers $k$ and $l$. Let $\cM^{0}:=\IZ_m^2=\lb 0,m-1\rb^2$ be the initial window area, i.e.  the support of the mask $\mu^{0}$. 
 Let $\cM$ be the object domain containing the support of the discrete object $f$. 
 %For convenience, we set  $\cM=\IZ_n^2$ for some integer $ n>m.$

Let $\cT$ be the set of all shifts, including $(0,0)$,  involved  in the ptychographic measurement. 
 Denote by $\mu^\bt$ the $\bt$-shifted probe for all $\bt\in \cT$ and $\cM^\bt$ the domain of
$\mu^\bt$. Let $f^\bt$ the object restricted to $\cM^\bt$.
\commentout{and  $\mb{\rm Twin}(f^\bt)$ the twin image of $f^\bt$ in $\cM^\bt$ defined as
\[
\mb{\rm Twin}(f^\bt)(\bn)=\bar f^\bt(\bN+2\bt-\bn),\quad\bn\in \cM^\bt,\quad\bN=(n,n).
\]
}
We  refer to each $f^\bt$ as a part of $f$ and write $f=\vee_\bt f^\bt$ {  where $\vee$ is the ``union" of functions consistent over their common support set}. In ptychography, the original object is broken up into a set of overlapping object parts, each of which produces a $\mu^\bt$-coded diffraction pattern.  
The totality of the coded diffraction patterns is called the ptychographic measurement data.  Let $\nu^{0}$ (with $\bt=(0,0)$) and $g=\vee_\bt g^\bt$ be any pair
 of the probe and the object estimates producing  the same ptychography data as $\mu^{0}$ and $f$, i.e.
 the diffraction pattern of $\nu^\bt\odot g^\bt$ is identical to that of $\mu^\bt\odot f^\bt$ where
 $\nu^\bt$ is the $\bt$-shift of $\nu^{0}$ and $g^\bt$ is the restriction of $g$ to $\cM^\bt$. 
 For simplicity, we assume  the periodic boundary condition on $\cM$ (i.e. discrete torus). 
The periodic boundary condition refers to the measurement scheme when the mask crosses over the boundaries of 
the object domain $\cM$ and 
should not be taken as 
the assumption of $f$ being a periodic object. The latter implies the former but not vice versa.

Consider the probe and object estimates
\beq
\label{lp1}
\nu^{0}(\bn)&=&\mu^{0}(\bn) \exp(-\im a -\im \br\cdot\bn),\quad\bn\in\cM^{0}\\
\label{lp2} g(\bn)&=& f(\bn) \exp(\im b+\im \br\cdot\bn),\quad\bn\in \cM
\eeq
for any $a,b\in \IR$ and $\br\in \IR^2$.  For any $\bt$, we have the following
calculation
\beqn
\nu^\bt(\bn)&=&\nu^{0}(\bn-\bt)\\
&=&\mu^{0}(\bn-\bt) \exp(-\im \br\cdot(\bn-\bt))\exp(-\im a)\\
&=&\mu^\bt(\bn) \exp(-\im \br\cdot(\bn-\bt))\exp(-\im a)
\eeqn
and hence for all $\bn\in \cM^\bt, \bt\in\cT$
\beq
\label{drift2}
\nu^\bt(\bn) g^\bt(\bn)&=&\mu^\bt(\bn)f^\bt(\bn) \exp(\im(b-a))\exp(\im \br\cdot\bt). 
\eeq
Since
for each $\bt$, $\nu^\bt\odot g^\bt$ is the phase  factor $ \exp(\im(b-a))\exp(\im \br\cdot\bt)$ times $\mu^\bt\odot f^\bt$ {  where $\odot$ is the entry-wise (Hadamard) product},  $g$ and $\nu^{0}$ produce the same ptychographic data as $f$ and $\mu^{0}$.  This  holds
true regardless of the set $\cT$ of shifts and the mask. 

In addition to the affine phase ambiguity \eqref{lp1}-\eqref{lp2},  a scaling factor ($g=c f, \nu^{0}=c^{-1} \mu^{0}, c>0$) is inherent to any blind ptychography. However,  when the mask is exactly known (i.e. $\nu^{0}=\mu^{0}$ up to a constant phase factor), $\br=0$ and $c=1$ so neither ambiguity can occur. 

In addition, for the regular raster scan (Fig. \ref{fig0}), it is well known that
blind ptychography is susceptible to many other artifacts \cite{probe09}.  For a complete analysis
of these ambiguities, the reader is referred to Ref. \cite{raster}.

A crucial question then is, Under what conditions are the scaling factor and the affine phase ambiguity 
the only ambiguities in blind ptychography? We aim to answer this question in this paper.

  Briefly and informally, we summarize the results as follows.

\subsection{Contributions}

 \begin{figure}[t]
\begin{center}
\subfigure[]{\includegraphics[width=5cm]{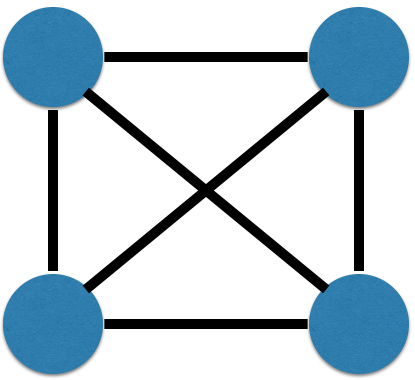}}\hspace{2cm}
\subfigure[]{\includegraphics[width=5cm]{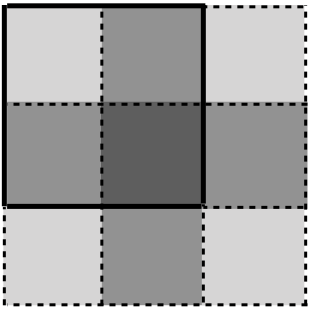}}
\caption{A complete undirected graph (a) representing four connected object parts (b) where the grey level indicates the number of coverages by the mask in four scan positions.
}
\label{fig:graph}
\end{center}
\end{figure}

\commentout{%07/10/2019

 Let the initial mask domain $\cM^0:=\IZ_m^2=\{(k,l): k,l=0,\cdots, m-1\}$ be the support of the mask function $\mu^0$.  Let $\cT$ be the set of all shifts, including $(0,0)$,  involved  in the ptychographic measurement.

 Denote by $\mu^\bt$ the $\bt$-shifted mask for all $\bt\in \cT$ and $\cM^\bt$ the domain of
$\mu^\bt$. Let $\cM:=\cup_{\bt\in \cT}\cM^\bt$.   Let $f^\bt$ the object restricted to $\cM^\bt$ and  $\mb{\rm Twin}(f^\bt)$ the twin image of $f^\bt$ in $\cM^\bt$. 
We can write $f=\vee_\bt f^\bt\subseteq\cM$ and refer to each $f^\bt$ as a part of $f$. In ptychography, the original object is broken up into a set of overlapping object parts, each of which produces a coded diffraction pattern (coded by $\mu^\bt$).  

The totality of the $\bt$-labelled coded diffraction patterns is called the ptychographic measurement data in this paper.  Let $\nu^0$ (with $\bt=(0,0)$) and $g=\vee_\bt g^\bt$ be any pair
 of the mask and the object estimates producing  the same ptychography data as $\mu^0$ and $f$, i.e.
 the diffraction pattern of $\nu^\bt\odot g^\bt$ is identical to that of $\mu^\bt\odot f^\bt$ where
 $\nu^\bt$ is the $\bt$-shift of $\nu^0$ and $g^\bt$ is the restriction of $g$ to $\cM^\bt$. 
 For convenience, $\mu^\bt, f^\bt,\nu^\bt, g^\bt $ assume the value zero outside of $\cM^\bt$. 
}

The first basic requirement of our method is the strong connectivity property of the object with respect to the measurement scheme. 
It is useful to think of connectivity in graph-theoretical terms (Fig. \ref{fig:graph}): Let the ptychographic experiment be represented by
a complete graph $\Gamma$ whose notes correspond to $\{f^\bt:\bt\in \cT\}$. Given any positive integer $s$, an edge between two nodes  corresponding to $f^\bt$ and $f^{\bt'}$ is $s$-connective if 
\beq
\label{r100}
%\red{What about $f^{\bt'}$?}
|\cM^\bt\cap \cM^{\bt'}\cap\supp(f)|\ge s
\eeq
where $|\cdot|$ denotes the cardinality.
In the case of full support (i.e. $\supp(f)=\cM$),  \eqref{r100} becomes
$|\cM^\bt\cap \cM^{\bt'}|\ge s$.
An $s$-connective reduced graph $\Gamma_s$ of $\Gamma$ consists
of all the nodes of $\Gamma$ but only the $s$-connective edges.  Two nodes are adjacent (and neighbors) in $\Gamma_s$ iff they are $s$-connected. A chain in $\Gamma_s$ is a sequence of nodes  such that 
two successive nodes are adjacent. In a simple chain all the nodes are distinct. Then the object parts $\{ f^\bt:\bt\in \cT\}$ are  $s$-connected if and only if $\Gamma_s$ is a connected graph, i.e. every two nodes is connected by a chain of $s$-connective edges. Loosely speaking, an object is strongly connected w.r.t. the ptychographic scheme if $s\gg 1$.

The second requirement is the existence of an {\em anchoring} part. Informally speaking, an object part $f^\bt$ is an anchor if
its support touches four sides of $\cM^\bt$ (Figure \ref{fig:corn}). Specifically, an object part $f^\bt$ is an anchor if  $f^\bt$ has a tight support in $\cM^\bt$, i.e. 
\beq
\label{tight}
\mb{\rm Box}[\supp(f^\bt)]=\cM^\bt
\eeq 
where $\mb{\rm Box}[E]$ stands for the box hull,  the smallest
rectangle containing $E$ with sides parallel to $\be_1=(1,0)$ or $\be_2=(0,1)$. An object part does not have a tight support if and only it has a loose support. Clearly, $f^\bt$ has a tight support if and only if $\mb{\rm Twin}(f^\bt)$ does  since  $\mb{\rm Box}[\supp(f^\bt)]=\mb{\rm Box}[\supp(\mb{\rm Twin}(f^\bt))]+\mbm$ for some $\mbm$. In the case $\supp(f)=\cM$, any object part is an anchor. 
For an extremely sparse object such as shown in Figure \ref{fig:corn}, the anchoring assumption can pose a challenge. 

 \begin{figure}[t]
\begin{center}
\includegraphics[width=12cm]{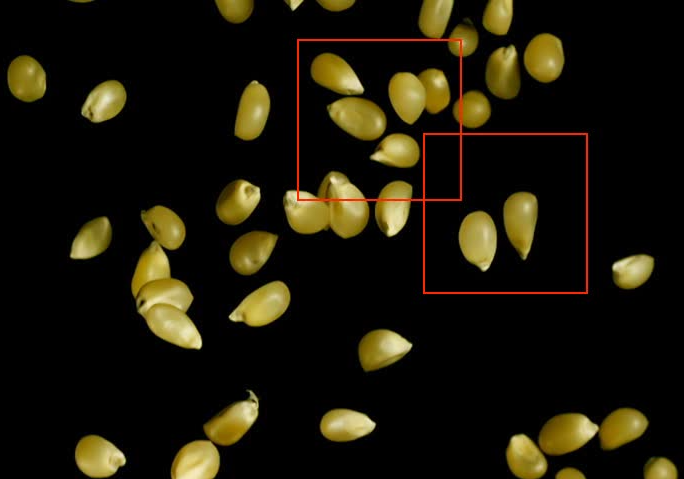}
\caption{Sparse objects such as this image of corn grains, where the dark area represents zero pixel value, can be challenging to ptychographic measurements.  The two red-framed blocks are not connected even though they overlap. The object part in the lower-right block is not an anchor since the object support does not touch the four sides of the block while
the object part in the upper-left block is an anchor. Indeed, the two corn grains at the lower-left and upper-right corners alone of the latter block
suffice to create a tight support. 
}
\label{fig:corn}
\end{center}
\end{figure}

For the unknown mask,  we need some prior information called the {\em mask phase constraint} (MPC):  
 \begin{quote}\em 
%Suppose that $\mu^0(\bn)$  are independently and continuously distributed nonvanishing random variables 
The mask estimate $\nu^0$ has  the property
$\Re(\bar \nu^0\odot \mu^0)>0$ at every pixel (where $\odot$ denotes the component-wise product and the bar denotes the complex conjugate). 
 \end{quote}
{   See Figure \ref{fig:MPC}. MPC can be relaxed as $|\arg[\nu^0(\bn)/\mu^0(\bn)]|<\pi/2$ for {\em sufficiently large percentage} of $\bn$. 
 For simplicity of presentation, however, we shall work with the technically simplifying version as above. 
 
Even with the perfect knowledge of the mask amplitude, MPC allows a large relative error 
\[
\sqrt{{1\over \pi}\int^{\pi/2}_{-\pi/2} |e^{\im\phi}-1|^2 d\phi}=\sqrt{2(1-{2\over \pi})}\approx 0.8525 %>\sqrt{{2\over 3}}.
\]
when $\arg[\nu^0]$ is selected randomly and uniformly in the interval $|\arg[\nu^0(\bn)/\mu^0(\bn)]|<\pi/2$. }

 \begin{figure}
 \centering
 \includegraphics[width=7cm]{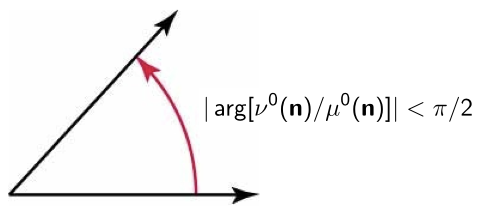}
\caption{$\nu^0$ satisfies MPC if $\nu_0(\bn)$ and $\mu^0(\bn)$ form an acute angle for all $\bn$.}
\label{fig:MPC}
\end{figure}

For any strongly connective scheme under the assumptions of MPC and anchoring, 
we prove the local uniqueness result for
blind ptychography  (Theorem \ref{thm:u} and \ref{thm:many}) that with high probability (exponentially close to 1 in $s$)  in  the random selection of $\mu^0$, 
\beq
\label{7.1}
\nu^{\bt}\odot g^\bt=e^{\im\theta_{\bt}}\mu^{\bt}\odot f^\bt,\quad \bt\in \cT, 
\eeq
for some constants $\theta_\bt\in \IR$ (called block phases)
 if $g$ and $\nu^\bt$ produce the same diffraction pattern as $f$ and $\mu^\bt$ for all $\bt\in \cT$. 
 As shown by  Examples \ref{ex0} and \ref{ex3.1}, both MPC and the anchoring assumption are in some sense necessary for \eqref{7.1} to hold. 

We refer to the ambiguity equation \eqref{7.1} as the {\em local uniqueness} property since $\theta_\bt$ may be more complicated than just an affine
 profile, $\theta_0+\bt\cdot\br,$  for some $\br\in \IR^2$,  as in \eqref{drift2}.  
Indeed, the affine phase ambiguity \eqref{lp1}-\eqref{lp2}  means that the relation \eqref{7.1} with
an affine profile in $\theta_\bt$ is the best to hope for. On the other hand, we say that the {\em global uniqueness}
holds if the affine phase ambiguity and the scaling factor ambiguity are the only ambiguities. We say that  a  ptychographic 
scheme is {\em complete} for a given object if the global uniqueness holds.

The ambiguity equation \eqref{7.1} can be transformed into the phase drift equation which plays the key role
in our theory. Consider the object ambiguity represented by
 \beqn
h(\bn)&\equiv& \ln g(\bn)-\ln f(\bn),\quad\forall \bn\in \cM, %&=&\im \theta_k -\ln \alpha^k(\bn-\bt_k)-\im \phi^k(\bn-\bt_k)\\
%\nn &=&\im \theta_{k}-\ln\alpha(\bn)-\im \phi(\bn) \mod \im 2\pi,\quad \forall \bn \in \cM^0,
\eeqn
provided that both $f$ and $g$ are non-vanishing. 
The phase drift equation 
 \beq
 \label{200.5}
 h(\bn+\bt)-h(\bn+\bt')=\im \theta_{\bt}-\im\theta_{\bt'} \mod \im2\pi,\quad \forall \bn\in \cM^{0},\quad \forall \bt,\bt'\in \cT
 \eeq
 equates the difference in the object ambiguity in different blocks with the phase drift in the block phase.

 \commentout{% 07/10/2019
{\bf Affine phase ambiguity.}
For any $\br\in \IZ^2, a,b\in \IR$, let  \beq
\label{lp1}
\nu^0(\bn)&=&\mu^0(\bn) \exp(-\im a-\im \br\cdot\bn),\quad\bn\in\cM^0\\
\label{lp2} g(\bn)&=& f(\bn) \exp(\im b+\im \br\cdot\bn),\quad\forall\bn\in \cM.
\eeq
For any $\bt$, we have the following
calculation
\beqn
\nu^\bt(\bn)&=&\nu^0(\bn-\bt)\\
&=&\mu^0(\bn-\bt) \exp(-\im \br\cdot(\bn-\bt))\exp(-\im a)\\
&=&\mu^\bt(\bn) \exp(-\im \br\cdot(\bn-\bt))\exp(-\im a).
\eeqn
Since, by the above calculation, 
\beq
\label{drift2}
\nu^\bt(\bn) g^\bt(\bn)&=&\mu^\bt(\bn)f^\bt(\bn) \exp(\im \br\cdot\bt)\exp(\im(b-a)),\quad \forall\bt\in \cT, \bn\in \cM^\bt, 
\eeq
$g^\bt$ and $\nu^\bt$ produce the same diffraction pattern as $f^\bt$ and $\mu^\bt$ for all $\bt$. 

In addition to the affine phase ambiguity \eqref{lp1}-\eqref{lp2}, another ambiguity,  a scaling constant factor ($g=c f, \nu^0=c^{-1} \mu^0, c>0$),
is also inherent to any blind ptychography as can easily be checked. 
}

Most important,   we show that the mixing schemes, introduced here for the first time, ``mix" the ambiguity so completely that 
a distinct  ambiguity  profile (affine phase plus scaling factor) arises and the global uniqueness holds true (Theorem \ref{thm:mix}). 
The mixing schemes include the special case of 
small perturbations of the regular raster scan (Theorems \ref{thm7.4} and \ref{thm7.5}). 
On the other hand, while the global uniqueness fails for  the regular raster scan, the block phases nevertheless have an affine profile 
(Proposition \ref{prop:raster}).

The rest of the paper is organized as follows.
In Section \ref{sec:ptych}, we formulate the basic building block of the
ptychographic measurement and  discuss ambiguities in standard phase retrieval
with one coded diffraction pattern. 
 In Section \ref{sec:two} we consider the ptychography with two overlapping diffraction patterns and 
prove the local uniqueness for the masked object (Theorem \ref{thm:u}). We then extend the local uniqueness to the multi-part ptychography (Theorem \ref{thm:many}). In Section \ref{sec:example} we demonstrate with examples that
the prior information of MPC and anchoring is necessary for the local uniqueness result (Examples \ref{ex0} and \ref{ex3.1}). In Section \ref{sec:block}, we develop the phase drift equation that holds the key to the global uniqueness result. In Section \ref{sec:raster}, we exhibit additional ambiguities associated
with the regular raster scan (Examples \ref{ex6} and \ref{ex5}) and prove that the block phases of the raster scan must have an affine
profile (Proposition \ref{prop:raster}).  In Section \ref{sec:simple}, we prove the global uniqueness theorems for
 the perturbed raster scans with the overlap ratio greater than $50\%$ (Theorems \ref{thm7.4} and \ref{thm7.5}). 
 In Section \ref{sec:blind}, we give an example showing that the minimum overlap ratio 50\% is necessary for the perturbed raster scans to be
 ptychographically complete 
 and introduce the mixing schemes which are ptychographically complete and whose block phases must have an affine profile (Theorem \ref{thm:mix}). 
We conclude in Section \ref{sec:con} and discuss a few practical implications of our theory. 
A preliminary version of this paper was presented in \cite{blind-ptycho}.

\section{Coded diffraction pattern}\label{sec:ptych}

We start with  the set-up of coded diffraction patterns \cite{Miao}. 

Let $f^0$ be a part of  the unknown object $f$ restricted to the initial block $\cM^0=\IZ^2_m, m<n,$ and let the Fourier transform of $f^0$ be written as 
\[
F(e^{-\im 2\pi\bw})=\sum_{\bk\in \cM^0} e^{-\im 2\pi \bk\cdot\bw} f^0(\bk),\quad \bw=(w_1,w_2). 
\]

Under the Fraunhofer 
approximation, the diffraction pattern  can be written as  \beq
|F(e^{-\im 2\pi\bw})|^2= \sum_{\bk \in \widetilde\cM^0}\lt\{\sum_{\bk'\in \cM^0} f^0(\bk'+\bk)\overline{f^0(\bk')}\rt\}
   e^{-\im 2\pi \bk\cdot \bom},\quad \bom\in [0,1]^2\label{auto}
   \eeq
   where
    \begin{equation*}
 \widetilde \cM^0 = \{ (k_1,k_2)\in \IZ^2: -m+1\le k_1 \le m-1, -m+1\le k_2\leq m-1 \} 
 \end{equation*}
 and $f^0$ assumes the value zero  outside of
 $\cM^0$. 
 Here and below the over-line notation means
complex conjugacy. 

The expression in the brackets  in \eqref{auto} is the autocorrelation function of $f^0$ and
the summation over $\bn$ takes the form of Fourier transform on  the enlarged  grid $\widetilde\cM^0$.
Hence sampling $|F|^2$ 
 on the grid 
\beq\label{L}
\cL = \Big\{(w_1,w_2)\ | \ w_j = 0,\frac{1}{2 m-1},\frac{2}{2m-1},\cdots,\frac{2m-2}{2m-1}\Big\}
\eeq
provides sufficient information to recover the autocorrelation function.

A randomly coded diffraction pattern measured with the mask $\mu^0$ is 
the diffraction pattern for 
the  {\em masked object} $
\tilde f^0(\bn) =f^0(\bn) \mu^0(\bn)$ 
where the mask function $\mu^0$ is a finite array of random variables.   
The masked object is also called the {\em exit wave} in the parlance of optics literature. 
In other words, a coded diffraction pattern is just the plain diffraction pattern of
a masked object. 

We assume  {randomness} in the phases $\theta$ of the mask function 
$
\mu^0(\bn)=|\mu^0|(\bn)e^{\im \theta(\bn)}
$
where  $\theta(\bn)$ are independent, continuous real-valued random variables. In other words, each $\theta(\bn)$ is independently distributed with a probability density function $p_\gamma$ supported on $(-\gamma\pi,\gamma\pi]$ with a constant $ \gamma\in [0,1]$. Continuous phase modulation can be experimentally realized with various
techniques such as spread spectrum phase modulation 
\cite{Spread}.

We also require that  $|\mu^0|(\bn)\neq 0,\forall \bn\in \cM^0$ (i.e. the mask is transparent). This is necessary for unique reconstruction of the object
as any opaque pixels of the mask would block the transmission of the object information. 
%By absorbing $|\mu^0(\bn)|$ into the object function we can assume, without loss of generality, that $|\mu^0(\bn)|=1,\forall\bn$, i.e.  a {\em phase} mask. With a proper choice of the normalizing constant  $c$, a phase mask then gives rise to an  {\em isometric}  

%\section{Ambiguities  in the masked object}\label{sec:one}

First we review the case of a plain diffraction pattern ($\mu^0\equiv 1$).  

\begin{prop}

\label{prop1}\cite{Hayes}
Let the $z$-transform 
$F(\bz) = \sum_{\bn} f^0(\bn) \bz^{-\bn}$ 
 be
given by
\beq
F(\bz)=\alpha \bz^{-\mbm} \prod_{k=1}^p F_k(\bz),\quad  \mbm\in \IN^2, \quad \alpha\in \IC\label{21}
\eeq
where $F_k, k=1,\dots,p,$ are 
non-monomial irreducible polynomials. Let $G(\bz)$ be
the $\bz$-transform of another finite array $g^0(\bn)$. 
%vanishing outside ${\mathbf 0}\leq \bn\leq \bN$.
Suppose 
$|F(e^{-\im 2\pi\bw})|=|G(e^{-\im 2\pi\bw})|,\forall \bw\in [0,1]^2$. Then \beq
\label{21'}
G(\bz)=|\alpha| e^{\im \theta} \bz^{-\bp}
\lt(\prod_{k\in I} F_k(\bz)\rt)
\lt(\prod_{k\in I^c} \overline{F_k(1/\bar{\bz})}\rt),\quad\mbox{for some}\quad \bp\in \IN^2,\, \theta\in \IR, 
\eeq
 where $I$ is a subset of $\{1,2,\dots,p\}$. 
\end{prop}

\begin{rmk}\label{rmk1}
The undetermined monomial factor $\bz^{-\bp}$ in \eqref{21'} corresponds to the 
translation invariance of the Fourier intensity data while 
the altered factors $\overline{F_k(1/\bar{\bz})}$ corresponds to the conjugate inversion invariance
of the Fourier intensity data (see Corollary \ref{cor1} below). 
The conjugate inversion of $f^0$, called the twin image, in $\cM^0$ is defined by  $\mb{\rm Twin} (f^0) (\bn)= \bar f^0((m,m)-\bn)$. 

\end{rmk}

Next consider  a random mask $\mu^0$ and assume that $f^0$ is not a linear  object. An object is a linear 
object  if its support  is a subset of  a line. 
We recall a result in \cite{unique} that the $z-$transform of the non-line masked object $\tilde f^0(\bn) =f^0(\bn) \mu^0(\bn)$
is irreducible, up to a monomial. 
\begin{prop}\cite{unique}
Suppose $f^0$ is not a linear object and let $\mu^0$ be the phase mask with phase at each point
 continuously and independently distributed. 
Then with probability one
the $z$-transform of the masked object $\tilde f^0 =f^0\odot \mu^0$ does not have any non-monomial irreducible
polynomial factor. 
\label{prop2}
\end{prop}
A similar result can be proved for masks whose phases are discrete random variables 
by using more advanced tools from algebraic geometry (e.g.
\cite{Bodin}, Proposition 4.1).

The following corollary is what we will need for proving the local uniqueness theorems. 
\begin{cor}\label{cor1} Under the assumptions of Proposition \ref{prop2}, if another masked object $\tilde g^0:=\nu^0g^0$  produces
the same diffraction pattern as $\tilde f^0=\mu^0f^0$, then for some $\bp$ and $\theta$
\beq\label{r1}
 \tilde f^0(\bn+\bp)&=& e^{-\im \theta}\tilde g^0(\bn) \quad\mbox{or}\quad
 e^{\im \theta}\, \mb{\rm Twin}(\tilde g^0)(\bn)
\eeq
for all $ \bn \in \cM^0$.
\end{cor}
\begin{proof} Let $\tilde F$ and $\tilde G$ be the $z$-transforms of $\tilde f^0$ and $\tilde g^0$, respectively. 
By Proposition  \ref{prop2} and \eqref{21'}, 
\beqn
\tilde G(\bz)=e^{\im \theta} \bz^{-\bp}
\tilde F(\bz)\quad\mbox{or}\quad e^{\im \theta} \bz^{-\bp}
\overline{\tilde F(1/\bar{\bz})},\quad\mbox{for some}\,\,\bp,\, \theta\,\,\mbox{and all}\,\,\bz. 
\eeqn
which after substituting $\bz=\exp{(-\im 2\pi \bw)}$ becomes 
\beqn
\tilde G(e^{-\im 2\pi\bw})=e^{\im \theta} e^{\im \bw\cdot\bp}
\tilde F(e^{-\im 2\pi\bw})\quad\mbox{or}\quad e^{\im \theta} e^{\im \bw\cdot\bp}
\overline{\tilde F(e^{-\im 2\pi\bw})},\quad\mbox{for some}\,\,\bp,\, \theta\,\,\mbox{and all}\,\,\bz. 
\eeqn

Note that $\tilde G(e^{-\im 2\pi\bw})$ and $\tilde F(e^{-\im 2\pi\bw})$ are the Fourier transforms of $\tilde g^0$ and
$\tilde f^0$, respectively. 
Therefore in view of Remark \ref{rmk1} we have 
\beqn
\tilde g^0(\bn)&=& e^{\im \theta} \tilde f^0(\bn - \bp)\quad\mbox{\rm or}\quad  e^{\im \theta}\, \mbox{\rm Twin}(\tilde f^0)(\bn-\bp),\quad\forall \bn \in \cM^0, 
\eeqn
which is equivalent to \eqref{r1}. 
\end{proof}

\section{Local uniqueness}\label{sec:two} \label{sec:multi}

First let us consider two-part ptychography where $\cM=\cM^0\cup \cM^\bt$.

We need 
two pieces of prior information: one on the mask phase and the anchoring assumption
on an object part.

 {\bf  Mask Phase Constraint (MPC):}
{\em  Let $\mu^0$ be a nonvanishing random  mask with phase at each pixel distributed 
continuously and independently according to a probability density function $p_\gamma$  nonvanishing in $(-\gamma \pi,\gamma\pi]$ with a constant $\gamma\le 1$. 

 Let 
\beq
\label{2mask}
\alpha(\bn)\exp[\im \phi(\bn)]=\nu^0(\bn)/\mu^0(\bn), \quad \alpha(\bn)>0,\quad \forall \bn \in \cM^0.
\eeq
We say 
that $\nu^0$ satisfies \mpc if, for all $\bn\in \cM^0$ and some constant $\phi_0$ 
 \beq
\label{u1}
&|\phi(\bn)-\phi_0|\le \delta\pi&  \mod 2\pi,
%&|\phi(\bn)-\phi_0|< \pi/2  & \mod 2\pi, \quad\mbox{\rm if}\quad \gamma> 1/2.\label{u2}
 \eeq
 where 
 \beq
 \label{u2}
\delta<\min\lt(\gamma,1/2\rt).
\eeq
}

The larger $\gamma$ is, the more phase diversity there is in the mask; the
larger $\delta$ is, the weaker the \mpc is as a constraint. When $\gamma>1/2$, \mpc can be written simply as
\beq
\label{mpc2}
\Re(\bar\nu^0(\bn)\mu^0(\bn))>0,\quad\forall \bn\in \cM^0.
\eeq

We demonstrate the necessity of \mpc in Example \ref{ex0}.

The following theorem gives sufficient conditions of the local uniqueness for 2-part ptychography. 

\begin{theorem}
\label{thm:u}
Let $f^0$ and $f^\bt$ be  a non-linear objects. Suppose that an arbitrary object $g=g^0\vee g^\bt$, where $g^0$ and $g^\bt$ are defined on $\cM^0$ and $\cM^\bt$, respectively,  and an arbitrary mask $\nu^0$ defined on $\cM^0$ produce the same ptychographic data as $f$ and $\mu^0$. Moreover, suppose that $\nu^0$ satisfies \mpc and that $f^0$ and  $g^0$ are an anchor, i.e.
\beq
\label{anchor2}
\mb{\rm Box}[\supp(f^0)]=\mb{\rm Box}[\supp(g^0)]=\cM^0.
\eeq

Let  \beq
 \label{3.2.1}
s=\min\{|S_0|, |S_0'|\}\ge 2
\eeq
where \[
S_0=\cM^0\cap \cM^\bt\cap \supp(f^0), \quad S_0'=\cM^0\cap \cM^\bt\cap \supp( \mb{\rm Twin}(f^0)).
\]

Then for some constants $\theta_0, \theta_\bt\in \IR,$ the following relations
 \beq
\label{masked1}
\nu^0\odot g^0&=&e^{\im\theta_0}\mu^0\odot f^0\\
\nu^\bt\odot g^\bt&=&e^{\im \theta_\bt}\mu^\bt\odot  f^\bt 
\label{masked2}
\eeq
hold true
with  probability at least 
\beq
\label{high}
1-c^s,\quad c<1, 
\eeq
where 
 the positive constant $c$ depends only on $\delta,\gamma, p_\gamma$ in \mpc. 
\end{theorem}
\begin{rmk}
The anchoring assumption  can be relaxed to that of object support constraint (OSC) (see Appendix \ref{sec:osc}). 
\end{rmk}
The proof of Theorem \ref{thm:u} is given in Appendix \ref{sec:proof}.

\commentout{
As commented below Assumption II, if  $\mb{\rm Box}[\supp(f^0)]=\cM^0$, then
OSC becomes null. So we have the following corollary. 
\begin{cor} \label{thm:u0} Suppose that
$f^0$  has a tight box hull in $\cM^0$. 
Then \eqref{masked1}-\eqref{masked2} hold with probability at least $1-c^s$, for some constant $ c\in (0,1)$,  with 
\beq
\label{sss}
s=|\cM^0\cap \cM^\bt\cap\supp(f^0)|\wedge |\cM^0\cap \cM^\bt\cap \supp(\mb{\rm Twin}(f^0))|.
\eeq
\end{cor}
}

%\section{Multi-part ptychography}\label{sec:multi}

Theorem \ref{thm:u} can be readily extended  to the case of multi-part ptychography as follows.

 Let $\cT=\{\bt_k \in \IZ^2: k=0,\dots,Q-1\}$ denote the set of all shifts in a ptychographic measurement. Let $\cM^k\equiv\cM^{\bt_k}$ and $f^k\equiv f^{\bt_k}$.
 
We say that $f^k$ and $f^l$ are $s$-connected if
\beq
\label{s-conn}
%\min\lt\{|\cM^k\cap \cM^{l}\cap\supp(f^k)|, |\cM^k\cap \cM^{l}\cap \supp(\mb{\rm Twin}(f^k))|\rt\}\ge s\ge 2
|\cM^k\cap \cM^{l}\cap\supp(f)|\ge s\ge 2
\eeq
(cf. \eqref{3.2.1}) and that $\{ f^k: k=1,\cdots, Q-1\}$ are $s$-connected if there is an $s$-connected chain between any
two elements.

\begin{thm}\label{thm:many} 
Let $\{f^k, k=0,\cdots,Q-1\} $ be $s$-connected and every $f^k$ is a non-linear part.
  %Let \beq\nn
%s_{\rm min}=\min_{k=0,\dots,q-2}\min\Big\{|\cM^k\cap \cM^{k+1}\bigcap\supp(f^k)|, |\cM^k\cap \cM^{k+1}\cap \supp(\mb{\rm Twin}(f^k))|\Big\}.  
%\eeq

Suppose that an arbitrary object $g=\bigvee_k g^k$, where $g^k$ are defined on  $\cM^k$,  and a mask $\nu^0$ defined on $\cM^0$ produce the same ptychographic data as $f$ and $\mu^0$.   Suppose that $\nu^0$ satisfies \mpc and hence
\beq
\label{pgamma}
p:=\max_{a\in \IR}\mbox{\rm Pr} \{\Theta\in (a-2\delta\pi, a+2\delta\pi]\}<1
\eeq
with $\Theta$ distributed according to the probability density function $p_\gamma \star p_\gamma.$

In addition, suppose that for some ${\ell_0}\in \{0,1,\dots,Q-1\}$
$f^{\ell_0}$ and $g^{\ell_0}$ are an anchor or more generally 
\beq
\label{3.101}
\nu^{\ell_0}\odot g^{\ell_0}=e^{\im\theta_{\ell_0}}\mu^{\ell_0}\odot f^{\ell_0}. 
\eeq
Then  with probability at least $1-2Qp^{s}$, we have
\beq
\label{3.100}
\nu^{k}\odot g^k=e^{\im\theta_{k}}\mu^{k}\odot f^k,\quad k=0,\dots,Q-1, 
\eeq
for some constants $\theta_k\in \IR.$ 

\end{thm}
\commentout{
\begin{rmk} For \eqref{3.101},  either the OSC (Theorem \ref{thm:u}) or
the tight box hull constraint (Corollary \ref{thm:u0}) on $f^{\ell_0}$ suffices. 

\end{rmk}
}
The proof of Theorem \ref{thm:many} is given in Appendix \ref{sec:many}. 

\commentout{The uniqueness expressed in  \eqref{3.100} is local in the sense that the undetermined  block phase can vary from block
to block arbitrarily, resulting in an uncontrollable phase drift as  
\eqref{3.100} relies entirely on
two-part connectivity.  In Sections \ref{sec:simple} and \ref{sec:blind}, we account for the interaction among multiple object parts
 to derive  global uniqueness. 
 }

\section{Ambiguities without  \mpc or anchoring assumption }\label{sec:example}

The first  example shows that  \eqref{masked1}-\eqref{masked2} may fail in the absence of \mpc.

\begin{ex}\label{ex0} Let $\cM=\IZ_m\times \IZ_n. $
Let $m=2n/3$ and $\bt=(m/2, 0)$. Evenly partition $f^0$ and $f^\bt$ into two  parts as 
$f^0=[f^0_0,f^0_{1}]$ and $f^\bt=[f^1_{0},  f^1_1]$ 
with the overlap $f^0_{1}=f^1_{0}$ where $f^i_j\in \IC^{m\times m/2},\, i,j=0,1$. Likewise, partition the mask as
$\mu^0=[\mu^0_0,\mu^0_1], \mu^\bt=[\mu^1_0,\mu^1_1]$  where $\mu^\bt$ is just the $\bt$-shift of $\mu^0$,
i.e. $\mu^\bt(\bn+\bt)=\mu^0(\bn)$. 

Suppose $f^0_0=f^1_1$ and consider the mask estimate $\nu^0=\mbox{\rm Twin}(\mu^0)$ and the following object estimate: Let 
\beqn
g^0&=&\mbox{\rm Twin}(f^0)=[g^0_0,g^0_1]\\
 g^\bt&=&\mbox{\rm Twin}(f^\bt)=[g^1_0,g^1_1]
 \eeqn
  where
$g^0_1=g^1_0$ due to $f^0_0=f^1_1$, i.e. $g=g^0\vee g^\bt$ is a well-defined object. 
The mask estimate $\nu^0$ violates  \mpc because 
\[
{\mbox{\rm Twin}(\mu^0)(\bn)\over \mu^0(\bn)}={\bar\mu^0(\bN-\bn)\over\mu^0(\bn)},\quad\bn\in\cM^0,
\]
has the maximum phase range $(-2\gamma\pi,2\gamma\pi]$. 

Clearly we have 
\beqn
\nu^0\odot g^0&=&\mbox{\rm Twin}(\mu^0\odot f^0)\\
\nu^\bt\odot g^\bt &=&\mbox{\rm Twin}(\mu^\bt\odot f^\bt)
\eeqn
so $\nu^0$ and $g$ produce the same ptychographic data as do $\mu^0$ and $f$ but
violate \eqref{masked1}-\eqref{masked2} since in general
\beqn
e^{\im\theta_0}\mu^0\odot f^0&\neq&\mbox{\rm Twin}(\mu^0\odot f^0)\\
e^{\im\theta_\bt}\mu^\bt\odot f^\bt &\neq&\mbox{\rm Twin}(\mu^\bt\odot f^\bt)
\eeqn
for any $\theta_0,\theta_\bt\in \IR$. 

\commentout{%09/05/2019
The same ptychographic ambiguity  persists for the different set-up with $m=n$ and $\bt=(m/2, 0)$
for the object parts $f^0=[f^0_0,f^0_{1}]$ and $f^\bt=[f^1_{0},  f^1_1]$ where the periodic boundary condition is imposed on $\IZ_n^2$. The periodic boundary condition implies
$f^0_0=f^1_1$ by definition. The above construction of the object and the estimates carries over here. 
}

\end{ex}

The  next example  illustrates the translational and twin-like ambiguities associated with a loose object support (non-anchor).  
\begin{ex}\label{ex3.1}
Assume the same set-up as in Example \ref{ex0} with the additional prior $f^0_0=f^1_1=0$. 

Let $
\nu^0=\mu^0, \nu^\bt=\mu^\bt$
 and $g^0=[g^0_{0}, 0], g^\bt=[0, g^1_{1}]$  where
\beqn
g^0_{0}&= &f^0_{1}\odot \mu^0_1/\mu^0_0,\\
g^1_{1}&=& f^1_{0}\odot \mu^1_0/\mu^1_1.
\eeqn
Clearly, $g=[g^0_{0}, 0, g^1_{2}]$ is different from $f=[0,f^0_{1}, 0]$.

It is straightforward to check that for $\mbm=(m/2,0)$
\beqn
g^0(\bn)\nu^0(\bn)=f^0(\bn+\mbm) \mu^0 (\bn+\mbm),\quad\bn\in \cM^0\\
g^\bt(\bn)\nu^\bt(\bn)=f^\bt(\bn-\mbm) \mu^\bt (\bn-\mbm),\quad\bn\in \cM^\bt
\eeqn
and hence $g^0\odot\mu^0$ and $ g^\bt \odot\mu^\bt$ produce  the same  diffraction patterns
as $f^0\odot\mu^0$ and $f^\bt\odot\mu^\bt$ for any $\nu^0$. In particular, by setting $\nu^0=\mu^0$, we satisfy MPC 
with $\delta=0$. 

On the other hand, for $\mbm\neq 0$ and any $\theta_0,\theta_\bt\in \IR$,  
\beqn
e^{\im\theta_0} f^0\odot \mu^0 &\neq& f^0(\cdot+\mbm) \odot \mu^0 (\cdot+\mbm)\\
e^{\im\theta_\bt} f^\bt\odot \mu^\bt&\neq & f^\bt(\cdot-\mbm) \odot \mu^\bt (\cdot-\mbm)
\eeqn
in general and  hence \eqref{masked1}-\eqref{masked2} are violated. 

For the twin-like ambiguity, consider the same set-up with
\beq
\label{3.26}
g^0(\bn)&= &\bar f^0(\bN-\bn)\bar\mu^0(\bN-\bn)/\mu^0(\bn),\quad\forall\bn\in \cM^0\\
\label{3.27}g_{\bt}(\bn)&= &\bar f^\bt(\bN+2\bt-\bn)\bar \mu^\bt(\bN+2\bt-\bn)/\mu^\bt(\bn), \quad\forall\bn\in \cM^\bt. 
\eeq
Clearly, $g=[g^0_{0}, 0, g^1_{2}]$ is different from  $f=[0,f^0_{1}, 0]$ but because  
\beqn
g^0(\bn)\nu^0(\bn)&=&\bar f^0(\bN-\bn) \bar \mu^0 (\bN-\bn),\quad\bn\in \cM^0\\
g^\bt(\bn)\nu^\bt(\bn)&=&\bar f^\bt(\bN+2\bt-\bn) \bar\mu^\bt (\bN+2\bt-\bn),\quad\bn\in \cM^\bt,
\eeqn
 $g^0\odot\mu^0$ and $ g^\bt \odot\mu^\bt$, as twin images,  produce  the same  diffraction patterns
as $f^0\odot\mu^0$ and $f^\bt\odot\mu^\bt$ for any $\nu^0$. In particular, by setting $\nu^0=\mu^0$, we satisfy MPC 
with $\delta=0$.

On the other hand, 
\eqref{masked1}-\eqref{masked2} fail to hold  since for any $\theta_0,\theta_\bt\in \IR$, 
\beqn
e^{\im \theta_0} f^0\odot \mu^0&\neq&\bar f^0(\bN-\cdot) \odot\bar \mu^0 (\bN-\cdot)\\
e^{\im\theta_\bt} f^\bt\odot\mu^\bt&\neq &\bar f^\bt(\bN+2\bt-\cdot) \odot \bar\mu^\bt (\bN+2\bt-\cdot)
\eeqn
in general.

\commentout{
Both above constructions violate \eqref{masked1}-\eqref{masked2} but satisfy the OSC \eqref{shifts2} with 
\[
T_0=\Big\{(a,0): a=0,\dots, m/2\Big\}.%\quad \cS'= \Big\{(a,0): a=-m/3,\dots, 0\Big\}  
\]
On the other hand, if  $f^{0}_{1}, f^{1}_{0}$ are non-vanishing, then it can be verified 
that 
$s=0$, consistent with
the fact that the probability for ambiguity is one as shown in the above construction. 

If we enhance the precision of the support knowledge by tightening  $T_0$ by any amount $l\ge 1$ as 
\beq
\label{3.30}
T_0= \Big\{(a,0): a=0,\dots, m/2-l\Big\},%\quad \cS'=\Big\{(a,0): a=-m/3+l,\dots, 0\Big\},\quad l\ge 1
\eeq
then the above constructions  would violate the OSC \eqref{shifts2}, and be rejected. Moreover, for \eqref{3.30}, $s=ml$ with nonvanishing $f^{0}_{1}, f^{1}_{0}$ so  the probability of uniqueness is closed to one  for
$m\gg 1$ as predicted by Theorem \ref{thm:u}. 
}
 
\end{ex}

\section{Phase drift equation}\label{sec:block}

In view of Theorem \ref{thm:many},  we make simple observations and transform \eqref{3.100} into the ambiguity equation  that will be a key to subsequent development. 
\begin{lemma}
\label{cor:u}
 Let 
\beqn
%\label{2mask'}
\alpha(\bn)\exp[\im \phi(\bn)]=\nu^0(\bn)/\mu^0(\bn), \quad \alpha(\bn)>0,\quad \forall \bn \in \cM^0
\eeqn
and 
\beqn
h(\bn)&\equiv& \ln g(\bn)-\ln f(\bn),\quad\forall \bn\in \cM, %&=&\im \theta_k -\ln \alpha^k(\bn-\bt_k)-\im \phi^k(\bn-\bt_k)\\
%\nn &=&\im \theta_{k}-\ln\alpha(\bn)-\im \phi(\bn) \mod \im 2\pi,\quad \forall \bn \in \cM^0,
\eeqn
where $f$ and $g$ are assumed to be non-vanishing.
%if   $g(\bn+\bt_{k})f(\bn+\bt_{k})\neq 0.$
%If $g(\bn+\bt_{k})=0$ and $f(\bn+\bt_{k})=0$, then $h(\bn+\bt_{k})$ may not be well defined.
Suppose that
\beq
\label{200.2}
\nu^{k}\odot g^k=e^{\im\theta_{k}}\mu^{k}\odot f^k,\quad \forall k, 
\eeq
where $\theta_k$ are constants.
Then
\beq \label{200.1}
h(\bn+\bt_{k}) &=&\im \theta_{k}-\ln\alpha(\bn)-\im \phi(\bn) \mod \im 2\pi,\quad\forall\bn\in \cM^0,\eeq
and 
for all $\bn\in  \cM^{k} \cap \cM^{l}$
\beq
\label{factor2}
\alpha(\bn-\bt_l)&=&\alpha(\bn-\bt_k)\\
 \theta_k- \phi(\bn-\bt_k)&= & \theta_l- \phi(\bn-\bt_l) \mod 2\pi.  \label{factor2'}
\eeq
%If, in addition,  $\phi(\bn-\bt)=\phi(\bn)$ for some $\bn\in  \cM^0 \cap \cM^\bt\cap \supp(f)$, then $\theta_0=\theta_\bt$. 
\end{lemma}
\begin{rmk} The ambiguity  equation \eqref{200.1} 
is a manifestation of local uniqueness \eqref{3.100} and  has the immediate consequence 
 \beq
 \label{200.4}
 h(\bn+\bt_{k})-h(\bn+\bt_{l})=\im \theta_{k}-\im\theta_{l} \mod \im2\pi,\quad \forall \bn\in \cM^{0},\quad \forall k,l
 \eeq
 or equivalently 
 \beq
 \label{200.5'}
 h(\bn+\bt_{k}-\bt_{l})-h(\bn)=\im \theta_{k}-\im\theta_{l} \mod \im2\pi,\quad \forall \bn\in \cM^{l}
 \eeq
 by shifting the argument in $h$. 

We refer to  \eqref{200.4}or \eqref{200.5'} as the {\em phase drift equation} which determines  the ambiguity 
(represented by $h$)  at different locations connected by ptychographic shifts. 
\end{rmk}
\begin{proof}
The ambiguity equation  \eqref{200.1} follows immediately from \eqref{200.2} by taking logarithm on both sides.

By \eqref{200.2},  for all $\bn\in \cM^k\cap \cM^l$,
\beq\label{r101}
g(\bn)={e^{\im\theta_k}f^k(\bn )\mu^0(\bn-\bt_k)/\nu^0(\bn-\bt_k)}
&=&e^{\im \theta_l}  f^l(\bn )\mu^0(\bn-\bt_l )/\nu^0(\bn-\bt_l). 
\eeq

We obtain by taking logarithm on both sides of \eqref{r101} that  \beqn
\im\theta_l-\im\theta_k  -\ln{f^k(\bn)}+\ln{f^l(\bn)}+\ln \alpha(\bn-\bt_k)-\ln \alpha(\bn-\bt_l)+\im\phi(\bn-\bt_k)-\im\phi(\bn-\bt_l)=0\nn
\eeqn
modulo  $\im 2\pi$. This implies 
that for $\bn\in  \cM^k \cap \cM^l$
\[
\im\theta_l-\im\theta_k +\ln \alpha(\bn-\bt_k)-\ln \alpha(\bn-\bt_l)+\im\phi(\bn-\bt_k)-\im\phi(\bn-\bt_l)=0\mod \im 2\pi\nn
\]
which is equivalent to \eqref{factor2}-\eqref{factor2'}. 

\end{proof}

% shall be set to maintain eq. \eqref{200.1} whose righthand side is always well-defined since $\alpha(\bn)\neq 0$ by assumption.

\commentout{%08/14/2019
\begin{prop}
\label{well-defined} Suppose that $\cT$ forms an additive group. Then 
the relation
\beq
\label{200.2'}
h(\bn+\bt_k)
=\im \theta_k-\ln\alpha(\bn)-\im \phi(\bn) \mod \im 2\pi,\quad \forall \bn \in \cM^0,
\eeq
can be maintained consistently under the assumptions of Theorem \ref{thm:many}.

\end{prop}
\begin{proof}
To show that \eqref{200.2'} can be maintained consistently when $g(\bn+\bt_k)=f(\bn+\bt_k)=0$, we need to establish that
\beq
\label{200.23}
\im \theta_1-\ln\alpha(\bn_1)-\im \phi(\bn_1) =\im \theta_2-\ln\alpha(\bn_2)-\im \phi(\bn_2) \mod \im 2\pi
\eeq
for any $\bn_1+\bs_1=\bn_2+\bs_2$ where $\bs_1,\bs_2\in \cT$. 

Noting  $\bn_2=\bn_1-\bt$ with $\bt=\bs_2-\bs_1\in \cT$, we have by Lemma \ref{cor:u} that 
\beq
\label{200.21}\alpha(\bn_2)&=&\alpha(\bn_1)\\
\phi(\bn_2)&=&\phi(\bn_1)+\theta_\bt-\theta_0\mod 2\pi.
\eeq
Since the block phases have an affine profile, $\theta_\bt-\theta_0=\theta_2-\theta_1$  and hence
\beq
\label{200.22}\phi(\bn_2)&=&\phi(\bn_1)+\theta_2-\theta_1\mod 2\pi.
\eeq

Together \eqref{200.21} and \eqref{200.22} imply \eqref{200.23}. 

%For any $\bn\in \cM^0$ suppose $f(\bn+\bt_k)\neq 0$ for some $k$. Then $h(\bn+\bt_k)$ is uniquely defined by \eqref{200.1}. 

\end{proof}

However, we are interested in more general schemes than those forming additive groups so
we will assume that $f$ does not vanish in $\cM$ for the rest of the paper.

}

\section{Raster scan}\label{sec:raster}

To fix the idea, we set $\cM=\IZ_n^2$ for the rest of the paper. 

Note that no other assumptions than the anchoring assumption and  the connectivity conditions, \eqref{3.2.1} and \eqref{s-conn},
are imposed on the scan scheme in Theorem \ref{thm:many}. In particular, Theorem \ref{thm:many} applies to the regular raster scan which 
 is more conveniently described  in terms of  two indices: 
 For some $q\in \IN$,
  \beq
 \label{raster}
\bt_{kl}=\tau(k,l)=k\tau\be_1+l\tau\be_2, \quad k,l=0,\dots, q-1,
\eeq
where $\be_1=(1,0), \be_2=(0,1)$ and $\tau$ is the constant step size of the raster scan. For simplicity of the set-up, we also assume that $\tau=m/p=n/q$ for some integers $p,q$  so that $\bt_{ql}=\bt_{0l},\bt_{kq}=\bt_{0l}$ and the periodic boundary condition on $\IZ_n^2$ is satisfied. 

We first show that the regular raster scan gives rise to an affine profile of block phase.

\begin{prop}\label{prop:raster}Under the assumptions of Lemma \ref{cor:u}, the block phase
$\{\theta_{kl}\}$ for the raster scan \eqref{raster} has an affine profile:
\beq\label{affine}
\theta_{kl}=\theta_{00}+r_1k+r_2l
\eeq
for $r_1, r_2\in \IR.$ 
\end{prop}
\begin{rmk}
Due to the affine phase ambiguity, $r_1$ and $r_2$ are undetermined constants.
\end{rmk}
\begin{proof} By \eqref{200.4},  for all $\bn\in \cM^{00}\cap (\cM^{00}-(\tau,0))$,
\beq\label{400}
h(\bn+(\tau,0))=h(\bn)+\im\theta_{10}-\im\theta_{00}
\eeq
and hence
\beq\label{402}
h(\bn+\bt_{kl})&=&h(\bn)+\im \theta_{kl}-\im\theta_{00}\\
&=&h(\bn+(\tau,0))+\im\theta_{kl}-\im\theta_{10}.\nn
\eeq
On the other hand,  \eqref{200.4} also implies 
\beq \label{401}
h(\bn+(\tau,0)+\bt_{kl})=h(\bn+(\tau,0))+\im \theta_{kl}-\im\theta_{00}
\eeq
and by \eqref{402} 
\beq\label{403}
h(\bn+(\tau,0)+\bt_{kl})&=&h(\bn+\bt_{kl})-\im\theta_{kl}+\im\theta_{10}+\im \theta_{kl}-\im\theta_{00}\\
&=& h(\bn+\bt_{kl})+\im\theta_{10}-\im\theta_{00}\nn
\eeq
for all $\bn\in \cM^{00}\cap (\cM^{00}-(\tau,0))$. 

By induction with \eqref{403}, we have
\beq
\label{405}h(\bn+(\tau,0)+\bt_{kl})&=& h(\bn+\bt_{0l})+(k+1)\im(\theta_{10}-\theta_{00}).
\eeq
Likewise, we also have
\beq\label{406}
h(\bn+(0,\tau)+\bt_{kl})&=& h(\bn+\bt_{k0})+(l+1)\im(\theta_{01}-\theta_{00}).
\eeq
Combining \eqref{405} and \eqref{406} with \eqref{200.4}, we arrive at the desired
result \eqref{affine} with 
\beqn
r_1=\theta_{10}-\theta_{00}, \quad r_2=\theta_{01}-\theta_{00}.
\eeqn  \end{proof}

\begin{cor}\label{cor:1} 
For the raster scan \eqref{raster} with $\tau=1$, we have
\beq
\label{471}
h(\bn)&=& h(0)+\im\bn\cdot(r_1,r_2) \mod \im 2\pi,\\
\label{472}
\phi(\bn)&=&\theta_{00}-\Im[h(0)]-\bn\cdot (r_1,r_2) \mod 2\pi\\
\alpha&=&e^{-\Re[h(0)]}\label{474'}\\
\theta_{kl}&=& \theta_{00}+kr_1+lr_2,\quad k,l=0,\cdots,n-1,  \label{473}
\eeq
for all $\bn\in \IZ^2_n$ and  some $r_1,r_2\in \IR$ where $\alpha$ is given in \eqref{2mask}.  

\end{cor}

\begin{proof}
Setting $\tau=1$ in  \eqref{405}-\eqref{406}, we have the identity \eqref{471}.  

By \eqref{200.1}, 
\beq \label{200.1'}
h(\bn+\bt) &=&\im \theta_{\bt}-\ln\alpha(\bn)-\im \phi(\bn) \mod \im 2\pi,\quad \forall\bt\in \cT.\eeq

With $\bt=\bt_{00}$, \eqref{200.1'} and   \eqref{471} imply \eqref{472} and \eqref{474'}.

The relation  \eqref{473} follows from  \eqref{200.1'} 
and 
\[
h(\bn+\bt)=h(0)+\im(\bn+\bt)\cdot(r_1,r_2) 
\] 
for any $\bt\in \cT.$ Note that the argument for \eqref{473} is an independent proof from Proposition \ref{prop:raster}. 

\end{proof}

The expressions \eqref{471} and \eqref{472}  correspond to the affine phase ambiguity
while  \eqref{474'} is the scaling factor ambiguity.

Even though the global uniqueness \eqref{471}-\eqref{473} is our goal but the raster scan with $\tau=1$ has too much redundancy and  is impractical.
On the other hand,  when $\tau>1$,  there are many additional ambiguities associated with the regular raster scan, posing substantial challenge to blind ptychographic reconstruction \cite{raster}.  Two of these ambiguities are illustrated below. 

The first example shows the ambiguity induced by the affine profile of the block phase \eqref{affine}.

%\medskip
\begin{ex}\label{ex6} For $q=3, \tau=m/2$, let
 \beqn
 f&=&
\lt[\begin{matrix}
f_{00}&f_{10}&f_{20}\\
f_{01}& f_{11}  & f_{21}   \\
f_{02} & f_{12} & f_{22}
\end{matrix}\rt]\\
 g&=&
\left[
\begin{matrix}
f_{00}  & e^{\im 2\pi/3} f_{10} & e^{\im 4\pi/3} f_{20}   \\
e^{\im 2\pi/3} f_{01} &e^{\im 4\pi/3} f_{11} &f_{21}  \\
e^{\im 4\pi/3} f_{02}  & f_{12}& e^{\im 2\pi/3}  f_{22}  
\end{matrix}
\right]
\eeqn
be  the object and its reconstruction, respectively, where $f_{ij}\in \IC^{n/3\times n/3}$.  Let 
 \beqn
\mu^{kl}=
\left[
\begin{matrix}
\mu^{kl}_{00}  & \mu^{kl}_{10}    \\
\mu^{kl}_{01}&\mu^{kl}_{11}
\end{matrix}
\right],\quad \nu^{kl}=
\left[
\begin{matrix}
\mu^{kl}_{00}  &e^{-\im 2\pi/3}  \mu^{kl}_{10}    \\
e^{-\im 2\pi/3} \mu^{kl}_{01}&e^{-\im 4\pi/3} \mu^{kl}_{11}
\end{matrix}
\right], 
\eeqn
$ k,l=0,1,2,$
be the $(k,l)$-th shift of the probe and estimate, respectively, where $\mu^{kl}_{ij}\in \IC^{n/3\times n/3}$. 

Let $f^{ij}$ and $g^{ij}$ be the part of the object and estimate masked  by $\mu^{ij}$ and $\nu^{ij}$, respectively. For example, we have
\[   f^{00}=
\lt[\begin{matrix}
f_{00}&f_{10}\\
f_{01}& f_{11}
\end{matrix}\rt],\quad  f^{10}=
\lt[\begin{matrix}
f_{10}&f_{20}\\
 f_{11}  & f_{21}
\end{matrix}\rt],\quad  f^{20}=
\lt[\begin{matrix}
f_{20}&f_{00}\\
f_{21}& f_{01}
\end{matrix}\rt]
\]
{  and likewise for other $f^{ij}$ and $g^{ij}$}. 
 It is easily seen that $\nu^{ij}\odot g^{ij}=e^{\im (i+j)2\pi/3} \mu^{ij}\odot f^{ij}.$
\end{ex}

The next example illustrates the periodic artifact called the raster grid pathology. 

\medskip
\begin{ex}\label{ex5} For $q=3, \tau=m/2$ and any  $\psi\in \IC^{{n\over 3}\times {n\over 3}}$, let 
 \beqn
 f&=&
\lt[\begin{matrix}
f_{00}&f_{10}&f_{20}\\
f_{01}& f_{11}  & f_{21}   \\
f_{02} & f_{12} & f_{22}
\end{matrix}\rt]\\
 g&=&
\left[
\begin{matrix}
e^{-\im \psi}\odot f_{00}  & e^{-\im \psi} \odot f_{10} & e^{-\im \psi} \odot f_{20}   \\
e^{-\im \psi}\odot f_{01} &e^{-\im\psi} \odot f_{11} & e^{-\im\psi}\odot f_{21}  \\
e^{-\im \psi}\odot  f_{02}  & e^{-\im\psi}\odot  f_{12}& e^{-\im \psi} \odot  f_{22}  
\end{matrix}
\right]
\eeqn
be  the object and its reconstruction, respectively, where $f_{ij}\in \IC^{n/3\times n/3}$.  Let 
 \beqn
\mu^{kl}=
\left[
\begin{matrix}
\mu^{kl}_{00}  & \mu^{kl}_{10}    \\
\mu^{kl}_{01}&\mu^{kl}_{11}
\end{matrix}
\right],\quad \nu^{kl}=
\left[
\begin{matrix}
e^{\im\psi}\odot \mu^{kl}_{00}  &e^{\im \psi} \odot \mu^{kl}_{10}    \\
e^{\im \psi}\odot \mu^{kl}_{01}&e^{\im \psi} \odot \mu^{kl}_{11}
\end{matrix}
\right], 
\eeqn
$ k,l=0,1,2,$
be the $(k,l)$-th shift of the probe and estimate, respectively, where $\mu^{kl}_{ij}\in \IC^{n/3\times n/3}$. 

Let $f^{ij}$ and $g^{ij}$ be the part of the object and estimate illuminated by $\mu^{ij}$ and $\nu^{ij}$, respectively (as
in Example \ref{ex6}). It is verified easily that $\nu^{ij}\odot g^{ij}= \mu^{ij}\odot f^{ij}.$
\end{ex}

The overlap ratio of above examples is $50\%$ (since $\tau=m/2$). However, the above construction of ambiguities can be
easily extended to the raster scan with any overlap ratio.  Moreover, all other 
ambiguities for blind ptychography with the raster scan can be shown to be the combinations of the above two types of ambiguity \cite{raster}.  

On the other hand, in the case $\tau=1$ ($q=n$), the ambiguity in Example \ref{ex6} is identical to  the affine phase ambiguity \eqref{lp1}-\eqref{lp2} while the ambiguity in Example \ref{ex5} becomes
the constant phase factor inherent to any phase retrieval. 

For the rest of the paper, we  develop an approach to  characterizing a more general class of scan schemes that enjoy the global uniqueness property \eqref{471}-\eqref{473}
by leveraging the phase drift equation \eqref{200.4}-\eqref{200.5'} more effectively.   
We refer to such schemes as {\em ptychographically complete} schemes.

\commentout{%08/14/2019
%For simplicity, we impose the periodic boundary condition  on the object domain $\IZ_n^2$. 
We have the following observation  (cf. Lemma \ref{cor:u}).
\begin{corollary}\label{cor:u2}
Under the assumptions of Theorem \ref{thm:many}, 
if, in addition,  $h(\bn)=h(\bn+\bt_k)$ for some $\bn\in \cM^0$, then $\theta_k=\theta_0$. 

\end{corollary}
}

\section{Motivating example:  perturbed raster scan}\label{sec:simple}
 \begin{figure}
\centering
\includegraphics[width=9cm,height=3cm]{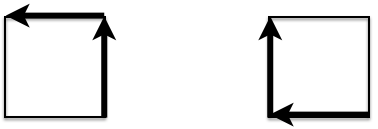}
\caption{Shortest paths (in the Manhattan distance) from the lower-right corner $(1,-1)$ to the upper-left corner $(0,0)$ in the diagrams spanned by $\bt_{kl}-\bt_{k-1,l}$ and $\bt_{k+1,l}-\bt_{kl}$. The left diagram corresponds to $\sigma_1$ in \eqref{reduce1} and the right diagram to $\sigma_2$ in \eqref{reduce2}.}
%{and the right diagram to $\sigma_3$ in \eqref{588}.  Many other paths  (not shown) are longer and  not constrained within the $3\times 3$ grid  most of which give rise to an empty $\Sigma_0$ and contribute nothing to \eqref{D55}. The dotted circles represent the grid points in the definition of $\Sigma_0$; they are all the grid points along the path except for the upper right corner/end points.  }
\label{fig2}
\end{figure}

Consider small perturbations to the raster scan:
\beq
\label{perturb}
\bt_{kl}=\tau (k,l)+(\delta^1_{kl},\delta^2_{kl}),\quad k,l=0,\dots,q-1
\eeq
where $\tau=n/q$, $\bt_{ql}=\bt_{0l},\bt_{kq}=\bt_{0l}$ (the periodic boundary condition) and $ \delta^1_{kl}, \delta^2_{kl}$ are small integers. 
Without loss of generality, we set $\delta^1_{00}=\delta^2_{00}=0$ and hence $\bt_{00}=(0,0)$ (see Fig. \ref{fig:irregular}(b)). 

We assume the {\em non-overstepping} condition that the perturbations do not change the ordering of $\{\bt_{kl}\}$, i.e.
\beq
\tau+\delta^1_{k+1,l}-\delta_{kl}^1>0,\quad  \tau+\delta^2_{k,l+1}-\delta_{kl}^2>0,  \quad k,l =0,\cdots, q-1.\label{530}
\eeq

Consider the triplet $(\bt_{k-1,l},\bt_{kl},\bt_{k+1,l})$ for any $k,l$ and let
\beq
\label{510}
\ba^1_{kl}:=(\bt_{kl}-\bt_{k-1,l})-(\bt_{k+1,l}-\bt_{kl})=2\delta^1_{kl}-\delta^1_{k-1,l}-\delta^1_{k+1,l}, 
\eeq
implying 
\beq
 &&h(\bn+  2\bt_{kl}-\bt_{k+1,l}-\bt_{k-1,l})=h(\bn+ \ba^1_{kl})\label{17.1}
\eeq
We want to reduce the lefthand side of \eqref{17.1} to $h(\bn)$ by using
\eqref{200.5'} repeatedly.

There are at least two paths for reduction:
\beq\label{reduce1}
\sigma_1: && (\bt_{kl}-\bt_{k-1,l})-(\bt_{k+1,l}-\bt_{kl})\longrightarrow \bt_{kl}-\bt_{k-1,l}\longrightarrow 0\\
\sigma_2:&&  (\bt_{kl}-\bt_{k-1,l})-(\bt_{k+1,l}-\bt_{kl})\longrightarrow -( \bt_{k+1,l}-\bt_{k,l})\longrightarrow 0\label{reduce2}
\eeq
corresponding to the two paths depicted in Fig. \ref{fig2}. 

Following $\sigma_1$, we have the identities
\beqn
h(\bn+\ba^1_{kl})&=& h(\bn+\bt_{kl}-\bt_{k-1,l})+\im\theta_{kl}-\im\theta_{k+1,l},\quad\forall \bn\in \cM^{kl}-\ba^1_{kl}\\
&=&h(\bn)+\im (2\theta_{kl}-\theta_{k-1,l}-\theta_{k+1,l})\quad \forall\bn\in \cM^{kl}-\bt_{kl}+\bt_{k-1,l}
\eeqn
implying 
\beq\label{550}
h(\bn+\ba^1_{kl})&=& h(\bn)+\im (2\theta_{kl}-\theta_{k-1,l}-\theta_{k+1,l})
\eeq
for all $\bn$ in the set
\beq
\lt[\cM^{kl}-\ba^1_{kl}\rt]\cap  \lt[\cM^{kl}-\bt_{kl}+\bt_{k-1,l}\rt].\label{551}
\eeq
On the other hand, following $\sigma_2$ we have the identities
\beqn
h(\bn+\ba^1_{kl})&=& h(\bn+\bt_{10}-\bt_{00})+\im\theta_{10}-\im\theta_{20},\quad\forall \bn\in \cM^{kl}-\ba^1_{kl}\\
&=&h(\bn)+\im (2\theta_{kl}-\theta_{k-1,l}-\theta_{k+1,l})\quad \forall\bn\in \cM^{kl}-\bt_{kl}+\bt_{k+1,l}
\eeqn
implying  \eqref{550} for all $\bn$ in the set
\beq
 \lt[\cM^{kl}-\ba^1_{kl}\rt]\cap  \lt[\cM^{kl}-\bt_{kl}+\bt_{k+1,l}\rt].\label{552}
\eeq
Combining the two routes of reduction, we have 
\beq
\label{200.10} \label{57'}
h(\bn+ \ba^1_{kl})=h(\bn)+\im ( 2\theta_{kl}-\theta_{k+1,l}-\theta_{k-1,l} )
\eeq
(modulo $\im2\pi$)
 for all $\bn$ in the set $(\cM^{kl}-\ba^1_{kl})\cap D^1_{kl}$ where
 \beq
 D^1_{kl}&:=& (\cM^{kl}-\bt_{kl}+\bt_{k-1,l}) \cup  (\cM^{kl}-\bt_{kl}+\bt_{k+1,l})\label{555}\\
 %&=& \lt[\cM^{kl}-(\tau+\delta^1_{k}-\delta^1_{k-1},0)\rt]\cup  \lt[\cM^{kl}+(\tau+\delta^1_{k+1}-\delta^1_{k},0)\rt].
 &=&\cM^{k-1,l}\cup\cM^{k+1,l}.\nn
 \eeq

Likewise, with
\beq
\ba^2_{kl}:= (\bt_{kl}-\bt_{k,l-1})-(\bt_{k,l+1}-\bt_{kl})=2\delta^2_{kl}-\delta^2_{k,l-1}-\delta^2_{k,l+1}\label{18.1}
\eeq
we have
 \beq
\label{prop2'}
h(\bn+ \ba^2_{kl})=h(\bn)+\im (2\theta_{kl}-\theta_{k,l+1}-\theta_{k,l-1})
\eeq
(modulo $\im2\pi$) for all $\bn$ in the set $D^2_{k,l}\cap (\cM^{k,l-1}-\ba^2_{kl})$ where
\beq
 D^2_{kl}&:= &(\cM^{kl}-\bt_{kl}+\bt_{k,l-1}) \cup  (\cM^{kl}-\bt_{kl}+\bt_{k,l+1})\label{556}\\
%  &=& \lt[\cM^{kl}-(0,\tau+\delta^2_{kl}-\delta^2_{k,l-1})\rt]\cup  \lt[\cM^{kl}+(0,\tau+\delta^2_{k,l+1}-\delta^2_{kl})\rt].\nn
&=& \cM^{k,l-1}\cup\cM^{k,l+1}. \nn
 \eeq

Repeatedly using \eqref{200.5'}, we can prove that the relation  \eqref{200.10} and \eqref{prop2'} hold respectively  in the sets
\beq
\label{59}
\bigcup_{\bt\in \cT}\lt[\bt-\bt_{kl}+(\cM^{k-1,l}\cup\cM^{k+1,l})\cap (\cM^{kl}-\ba^1_{kl})\cap\cM^{kl}\rt]
\eeq
and
\beq 
\label{58}\bigcup_{\bt\in \cT}\lt[\bt-\bt_{kl}+( \cM^{k,l-1}\cup\cM^{k,l+1})\cap (\cM^{kl}-\ba^2_{kl})\cap\cM^{kl} \rt] 
\eeq
 where the additional restriction due to the presence of $\cM^{kl}$
is to ensure the validity of applying \eqref{200.5'} (See Lemma  \ref{lem8.1} for a proof in a more general setting). 

For a special class of perturbed raster scans, precise conditions for the sets in \eqref{59}-\eqref{58} to cover $\IZ_n^2$ can be simply stated as follows.  

\begin{figure}[t]
\centering
\subfigure[Perturbed grid given by \eqref{perturb2}]{\includegraphics[width=5cm,height=4.7cm]{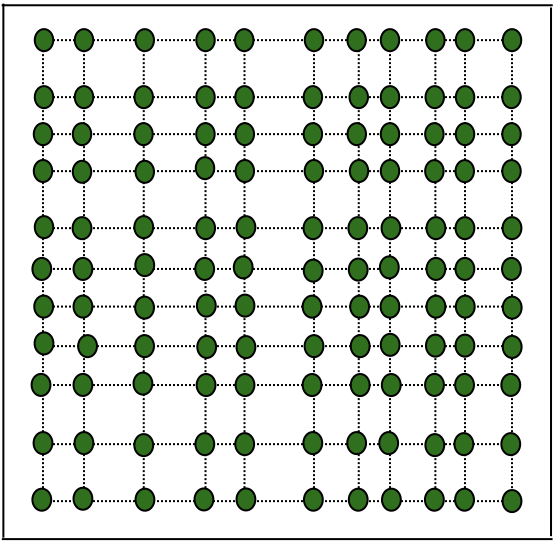}}\hspace{2cm}
\subfigure[Perturbed grid given by \eqref{perturb}]{\includegraphics[width=5cm,height=4.7cm]{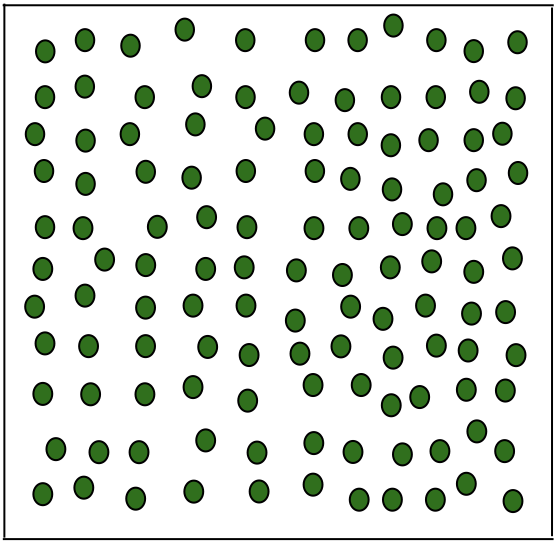}}
\caption{Two perturbed raster scans}
\label{fig:irregular}
\end{figure}

\begin{lem}\label{lem7.1} For the perturbed raster scan \eqref{perturb} with the non-overstepping condition \eqref{530}, suppose 
\beq
\label{perturb2}
\delta^1_{kl}=\delta^1_k, \quad \delta^2_{kl}=\delta^2_l,\quad\forall k,l=0,\cdots, q-1, 
\eeq
(Consequently, $\ba^1_{kl}=\ba^1_k,\ba^2_{kl}=\ba^2_l$), see Fig. \ref{fig:irregular}(a).

If for
some fixed $k,l,$
 \beq
\label{cover2}
2\tau \le m+\max\{\delta^1_{k-1}-\delta^1_{k+1}, \delta^2_{l-1}-\delta^2_{l+1}\}
\eeq
and 
\beq
\label{small1}
\label{small2}
&&\max_{i=1,2}[|a^i_{k}| +\max_{k'}\{\delta^i_{k'+1}-\delta^i_{k'}\}]\le m-\tau,
%0< e^1_{k}\le \tau+ \delta^1_{k}-\delta_{k-1}^1 \quad \mbox{or}\quad 0<-e^1_k\le \tau+\delta^1_{k+1}-\delta^1_{k}\\
%0< e^2_{l}\le \tau+ \delta^2_{l}-\delta_{l-1}^2 \quad \mbox{or}\quad 0<-e^2_l\le \tau+\delta^2_{l+1}-\delta^2_{l}.
\eeq
where 
\[
a^1_k=2\delta_k^1-\delta_{k-1}^1-\delta_{k+1}^1, \quad a^2_l=2\delta_l^2-\delta_{l-1}^2-\delta_{l+1}^2, 
\]
then each set in \eqref{59} and \eqref{58} 
\commentout{
 \beq
 \label{57}
h(\bn+(1,0))
&=&h(\bn)+\im r_1\mod \im2\pi
\eeq
}
 contains $\IZ_n^2$. 
\end{lem}
\begin{rmk}\label{rmk7.2}
For the raster scan \eqref{raster}, $a_k^1=a^2_l=0$ for all $k,l$. 

For small perturbations $\delta^1_k,\delta^2_l\ll 1$,  \eqref{small2} is satisfied and  \eqref{cover2} means an overlap ratio slightly greater than $50\%$. 
This is an improved and simplified version of  the one given in \cite{raster}. 

\end{rmk}
\begin{proof}
First \eqref{cover2} implies that the right edge of $\cM^{k-1,l}$ is no less than the left edge of $\cM^{k+1,l}$ by more than one pixel 
and that the upper edge of $\cM^{k-1,l}$ is no less than the lower edge of $\cM^{k,l+1}$ by more than one pixel. Hence
both $\cM^{k-1,l}\cup\cM^{k+1,l}$ and $ \cM^{k,l-1}\cup\cM^{k,l+1}$ are rectangles and by the non-overstepping condition \eqref{530} 
\beqn
\cM^{k-1,l}\cup\cM^{k+1,l} \supseteq \cM^{kl},\quad \cM^{k,l-1}\cup\cM^{k,l+1}\supseteq \cM^{kl}.
\eeqn

For the remaining  argument, it suffices to show that 
\beq\label{680}
\IZ_n^2\subseteq\bigcup_{\bt\in \cT}\lt[\bt-\bt_{kl}+ (\cM^{kl}-\ba^1_{k})\cap\cM^{kl} \rt],&& \IZ_n^2\subseteq\bigcup_{\bt\in \cT}\lt[\bt-\bt_{kl}+ (\cM^{kl}-\ba^2_{l})\cap\cM^{kl} \rt]. 
\eeq
To this end, since the intersection of two adjacent sets in \eqref{680} 
 \beq
 \label{710}
&&{\lt\{ \bt_{ij}-\bt_{kl}+\cM^{kl}\cap (\cM^{kl}-\ba^1_{k})\rt\}\cap \lt\{ \bt_{i+1,j}-\bt_{kl}+\cM^{kl}\cap (\cM^{kl}-\ba^1_{k}) \rt\}}\\
&& \lt\{\bt_{ij}-\bt_{kl}+\cM^{kl}\cap (\cM^{kl}-\ba^2_{l})\rt\}\cap \lt\{ \bt_{i,j+1}-\bt_{kl}+\cM^{kl}\cap (\cM^{kl}-\ba^2_{l})  \rt\}\label{712}
\eeq
are congruent to 
\beqn
&&{\lt\{\cM^{00}\cap (\cM^{00}-(a_k^1,0))\rt\}\cap \lt\{ (\tau+\delta^1_{i+1}-\delta^1_i,0)+\cM^{00}\cap (\cM^{00}-(a^1_k,0)) \rt\}}\\
&& \lt\{\cM^{00}\cap (\cM^{00}-(0,a^2_l))  \rt\}\cap \lt\{ (0,\tau+\delta^2_{j+1}-\delta^2_j)+\cM^{00}\cap (\cM^{00}-(0,a^2_l)) \rt\}, 
\eeqn 
\eqref{small1} implies that  neither set in \eqref{710}-\eqref{712} is empty for any $ i,j$. 
Therefore  \eqref{680} holds true. 
 \commentout{is no less than that of $\cM^{k-1,l}$ and
that the lower edge of $\cM^{kl}-\ba^2_l$ is no less than that of $\cM^{k,l-1}$;
in the meantime, the right edge of $\cM^{kl}-\ba^1_{k}$ is no greater than that of $\cM^{k+1,l}$
and the upper edge of $\cM^{kl}-\ba^2_l$ is no greater than that of $\cM^{k,l+1}$. 
}
\end{proof}

The following is an immediate consequence of \eqref{57'}, \eqref{prop2'} and Lemma \ref{lem7.1}. 

\begin{cor} Suppose that $f$ does not vanish in $\IZ_n^2$.  Under the assumptions of Lemma \ref{lem7.1}, if 
\beq
\label{711}
a^1_k=1, \quad a^2_l=1, \quad \mbox{for some $k,l$,}
\eeq
 then the scheme is ptychographically complete, i.e.
\beq\label{101}
h(\bn)&=&h(0)+\im \bn\cdot (r_1,r_2)\mod \im 2\pi\\
\label{472'}
\phi(\bn)&=&\theta_{00}-\Im[h(0)]-\bn\cdot(r_1,r_2) \mod 2\pi\\
\alpha&=&e^{-\Re[h(0)]}\label{474}\\
\theta_{\bt}&=&\theta_{00}+\bt\cdot (r_1,r_2)\mod 2\pi,\quad \bt\in \cT,  \label{473'}
\eeq
for all $\bn\in \IZ_n^2$ where  $r_1,r_2\in \IR$ are undetermined  constants (due to the affine phase ambiguity)
and $\alpha$ is given in \eqref{2mask}
\end{cor}

\begin{proof}
The assumption \eqref{711}, \eqref{57'}, \eqref{prop2'} and Lemma \ref{lem7.1} imply that
\[
h(\bn+\be_1)=h(\bn)+\im (2\theta_{kl}-\theta_{k+1,l}-\theta_{k-1,l}),\quad h(\bn+\be_2)=h(\bn)+\im (2\theta_{kl}-\theta_{k,l+1}-\theta_{k,l-1})
\]
\mbox{for all $\bn$ in $\IZ_n^2$} and hence \eqref{101}. 

The rest of the proof is exactly the same as that of Corollary \ref{cor:1}. In particular, \eqref{473'} follows from \eqref{101} and the phase drift equation \eqref{200.4}-\eqref{200.5'}. 
\end{proof}

More generally, we  have the  following global uniqueness theorem for the perturbed raster scan  \eqref{perturb2}. 

\begin{thm}\label{thm7.4} Suppose that $f$ does not vanish in $\IZ_n^2$. For the perturbed raster scan \eqref{perturb2} satisfying the non-overstepping condition \eqref{530} let $\{(\delta^1_{k_i},\delta^2_{l_j}):i,j\}$ be any nonempty subset of  perturbations  satisfying \eqref{cover2} and \eqref{small2} in Lemma \ref{lem7.1}. 

Let \[
a^1_{i}=2\delta^1_{k_i}-\delta^1_{k_i-1}-\delta^1_{k_i+1}, \quad
a^2_{j}=2\delta^2_{l_j}-\delta^2_{l_j-1}-\delta^2_{l_j+1}, \quad\forall i,j,
\]
 and suppose  \beq
\label{coprime0}
\gcd_{i}\lt(|a^1_{i}|\rt)=\gcd_{j}\lt(|a^2_{j}|\rt)=1 
\eeq
where \mbox{\rm gcd} denotes the greatest common divisor.  
\commentout{
 \beq
\label{cover2'}
&&2\tau\le m-\max_{i=1,2}\{\delta^i_{k_j+1}-\delta^i_{k_j-1}\}\\
\label{small2'}
&&\max_{i=1,2}[|a^i_{k_j}| +\max_{k'}\{\delta^i_{k'+1}-\delta^i_{k'}\}]\le m-\tau
%\label{small1'} &&\delta^i_{k_j+1}-\delta^i_{k_j+2}\le \tau\le m-1+\delta^i_{k_j+1}-\delta^i_{k_j+2}.
\eeq
}
Then the global uniqueness 
\eqref{101}-\eqref{473'} holds true and the scheme is ptychographically complete.  \end{thm}
\begin{proof}
The coprime condition \eqref{coprime0} implies the existence of  $c^1_{i},c^2_{j}\in \IZ$ such that
\beq
\sum_{i}c^1_{i}a^1_{i}=\sum_{j}c^2_{j} a^2_{j}=1.\label{prime}\label{coprime}
\eeq

%The non-overstepping condition \eqref{530} implies that  $\cM^{kl}$ overlaps with $\cM^{k+1,l}, \cM^{k,l+1},  \cM^{k-1,l}$ and $\cM^{k,l-1}$ for all $k,l$. 
\commentout{
Let 
\beq
\label{D1}
\Delta^1_{kl}:=2\theta_{k+1,l}-\theta_{k+2,l}-\theta_{kl}\\
\Delta^2_{kl}:=2\theta_{k,l+1}-\theta_{k,l+2}-\theta_{kl}. \label{D2}
\eeq
}

%(This observation is used to simplify the notation but not the essence of the remainder argument). 

By repeatedly using \eqref{57'} and \eqref{prop2'} we have 
\beqn
h(\bn+\be_1)&=&h\lt(\bn+(\sum_{i}c^1_{i} a^1_{i},0)\rt)\nn
= h(\bn)+\im r_1 \mod \im 2\pi\\
h(\bn+\be_2)&=&h\lt(\bn+(0,\sum_{j}c^2_{j} a^2_{j})\rt)
=h(\bn)+\im r_2 \mod \im 2\pi
\eeqn
where
\beqn
r_1=\sum_{i} c^1_{i} (2\theta_{k_i,i}-\theta_{k_i+1,i}-\theta_{k_i-1,i}),\quad r_2=\sum_{j} c^2_{j} (2\theta_{i, l_j}-\theta_{i,l_j+1}-\theta_{i, l_j-1}) \label{rr'}
\eeqn
and hence  \eqref{101}. 
%As an unintended consequence the sums in \eqref{rr'} that define $r_1$ and $r_2$ must be independent of $i$. 

\end{proof}

\commentout{%08/26/2019
The opposite case to \eqref{perturb2} is the perturbed raster scan with
\beq
\label{perturb3}
\delta^1_{kl}=\delta^1_l,\quad \delta^2_{kl}=\delta^2_k.
\eeq

}

Instead of linear shifts with uneven step sizes in \eqref{perturb2}, the general case \eqref{perturb} produces curvilinear shifts which is
more difficult to analyze. 
To state the analogous theorem for the general case  \eqref{perturb}, 
let  $ \bu_i:=(u_{i1},u_{i2}), i=1,2,$ be a $\IZ^2$-lattice basis, i.e. the four integers $u_{11}, u_{12}, u_{21}, u_{22} $ satisfy 
\beq
\label{basis}
u_{11}u_{22}-u_{12}u_{21}=1. 
\eeq
Since $u_{11}u_{22}-u_{12}u_{21}=1$, there exist integers $b_{ij},i,j=1,2,$ such that 
\[
b_{11} \bu_1+b_{12}\bu_2=\be_1=(1,0),\quad
b_{21}\bu_1+b_{22}\bu_2=\be_2=(0,1). 
\]

\begin{thm}\label{thm7.5}
Suppose that $f$ does not vanish in $\IZ_n^2$. For the perturbed raster scan \eqref{perturb} satisfying the non-overstepping condition \eqref{530}, let 
 $\{(\delta^1_{k_il_i},\delta^2_{k_jl_j}): i,j\}$ 
be any nonempty subset of perturbations such that 
\beq
\label{802}
\IZ^2_n\subseteq \bigcup_{\bt\in \cT}\lt[\bt-\bt_{k_il_i}+(\cM^{k_i-1,l}\cup\cM^{k_i+1,l})\cap (\cM^{k_il_i}-\ba^1_{k_i})\cap\cM^{k_il_i}\rt],\quad\forall i
\eeq
and
\beq 
\label{803}\IZ_n^2\subseteq \bigcup_{\bt\in \cT}\lt[\bt-\bt_{k_jl_j}+( \cM^{k_j,l_j-1}\cup\cM^{k_j,l_j+1})\cap (\cM^{k_jl_j}-\ba^2_{l_j})\cap\cM^{k_jl_j} \rt],\quad \forall j.
\eeq
Let \beq
\label{800}
\ba^1_{i}&:=&(\bt_{k_il_i}-\bt_{k_i-1,l_i})-(\bt_{k_i+1,l_i}-\bt_{k_il_i}),\quad\forall i\\
\ba^2_{j}&:=&(\bt_{k_jl_j}-\bt_{k_j,l_j-1})-(\bt_{k_j,l_j+1}-\bt_{k_jl_j}),\quad \forall j\label{801}
\eeq
and suppose that 
 \beq\label{805}
 \sum_i c^1_i \ba^1_i = \bu_1,\quad  \sum_j c^2_j \ba^2_j= \bu_2
 \eeq
 for some  $c_i^1, c^2_j\in \IZ$ where $\{\bu_1,\bu_2\}$ is a $\IZ^2$-lattice basis.
 Then the global uniqueness
\eqref{101}-\eqref{473'} holds true and the scheme is ptychographically complete.  
\end{thm}
\begin{rmk}
The conditions \eqref{802}-\eqref{803} are tedious to state in terms of the perturbations $\delta^1_{kl}, \delta^2_{kl}$ and
do not provide much insight beyond what is given in Remark \ref{rmk7.2}. 

\end{rmk}
\begin{proof}
As before, we begin with \beqn
h(\bn+\ba^1_i)&=&h(\bn)+ \im (2\bt_{k_il_i}-\bt_{k_i-1,l_i}-\bt_{k_i+1,l_i})\\
h(\bn+\ba^2_j)&=& h(\bn)+\im (2\bt_{k_jl_j}-\bt_{k_j,l_j-1}-\bt_{k_j,l_j+1})
\eeqn
$(\mod \im 2\pi)$ for all $\bn\in \IZ_n^2$
and repeatedly use  \eqref{805} to obtain
\beqn
h(\bn+\bu_1)=h(\bn)+\im\Delta_1,\quad 
h(\bn+\bu_2)=h(\bn)+\im\Delta_2
\eeqn
where
\beqn
\Delta_1&= &\sum_{i} c^1_i(2\theta_{k_il_i}-\theta_{k_i+1,l_i}-\theta_{k_i-1,l_i})\\ 
\Delta_2&=& \sum_{j} c^2_j(2\theta_{k_jl_j}-\theta_{k_j, l_j+1}-\theta_{k_j, l_j-1}). 
\eeqn

Since $u_{11}u_{22}-u_{12}u_{21}=1$, there exist integers $b_{ij},i,j=1,2,$ such that 
\[
b_{11} \bu_1+b_{12}\bu_2=\be_1,\quad
b_{21}\bu_1+b_{22}\bu_2=\be_2. 
\]
Therefore, for $j=1,2$,
\beqn
{h(\bn+\be_1)}
&=&h(\bn)+ \im b_{11}\Delta_1
+ \im b_{12}\Delta_2 \\
h(\bn+\be_2)&=& h(\bn)+ \im b_{21}\Delta_1 +\im  b_{22}\Delta_2, 
\eeqn
and  \eqref{101}-\eqref{473'} hold true.
\end{proof}

\section{Mixing schemes with three-part coupling}\label{sec:blind}

Let us begin with a simple example showing that a perturbed scan with overlap ratios less than 50\% may result in excessive ambiguities.

\begin{ex} \label{ex50} 
\begin{figure}
\centering
\includegraphics[width=8cm]{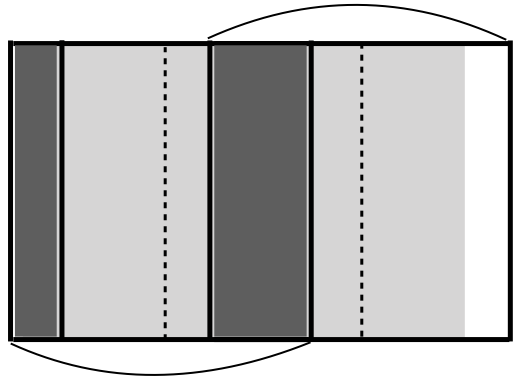}
\caption{A perturbed scan with $q=2$.  The arcs indicate the extend of the two blocks $\cM^{00}$ and $\cM^{10}$. The dotted lines mark the midlines of the two blocks. The grey area represents the object with the light grey areas being $R_{00}$ and $R_{10}$ and the dark grey areas being the overlap of the two blocks. The white area inside $\cM^{10}$ folds into the other end inside $\cM^{00}$ by the periodic boundary condition.}
\label{fig50}
\end{figure}

Let us consider the perturbed scheme \eqref{perturb2} with $q=2$ and  
\beq
\bt_{kl}=(\tau_k,\tau_l),\quad k,l=0,1,2
\eeq where $\tau_0=0,  \tau_2=n$ and 
\beq
{3m/2}<n<m+\tau_1. \label{84.0}
\eeq
The condition \eqref{84.0} is to ensure that the overlap ratio ($2-n/m$) between two adjacent blocks is less than (but can be made arbitrarily close to) $50\%$. To avoid the raster scan (which has many undesirable ambiguities \cite{raster}), we assume that $\tau_1\neq n/2$ and hence $\tau_2\neq 2\tau_1$. 
Note that the periodic boundary condition implies that $\cM^{00}=\cM^{20}=\cM^{02}=\cM^{22}$. 
Figure \ref{fig50} illustrates the relative positions of $\cM^{00}$ and $\cM^{10}$. 

First let us focus on the horizontal shifts $\{\bt_{k0}: k=0,1,2\}$. As shown in Figure \ref{fig50}, two subsets  of $\cM=\IZ^2_n$
\[
R_{00}= \lb m+\tau_1-n,\tau_1-1 \rb\times \IZ_m,\quad R_{10}= \lb m, n-1\rb\times \IZ_m
\]
 are covered only once  by $\cM^{00}$ and $\cM^{10}$
respectively due to the \eqref{84.0}.  It is straightforward to check that the conclusion of Lemma \ref{lem7.1} fails in this case. 

 %Clearly,  $R_0+\bt_{20}=R_1$. 
 Now consider the intersections 
 \beqn
 \tilde R_{10}:=R_{10}\cap (\bt_{10}+R_{00})&=&R_{10}\cap \lb m+2\tau_1-n,2\tau_1-1\rb\times \IZ_m\\
 \tilde R_{00}:= (R_{10}-\bt_{10})\cap R_{00}&=&\lb m-\tau_1, n-\tau_1-1\rb\times \IZ_m\cap R_{00}
 \eeqn
which respectively  correspond to the same region of the mask in $\cM^{10}$ and $\cM^{00}$
and let $h_1$ be any function defined on $\cM$ such that $h_1(\bn)=0$ for any $\bn\neq  \tilde R_{10}\cup \tilde R_{00}$
and $h_1(\bn+\bt_{10})=h_1(\bn)$ for any $\bn\in \tilde R_{00}$. 

Consider the object estimate
$ g(\bn)=
e^{h_1(\bn)}f(\bn)$ and 
 the mask estimate
$
\nu^{k0}(\bn):=e^{-h_1(\bn)}\mu^{k0}(\bn) $,  which is well defined because $\tilde R_{10}=\bt_{10}+\tilde R_{00}$ and
both correspond to the same region of the mask.

By the same token, we can construct a similar ambiguity function $h_2$ for the vertical shifts. 
With both horizontal and vertical shifts, we define the ambiguity function $h=h_1h_2$ and the
associated pair of mask-object estimate $
\nu^{kl}(\bn):=e^{-h(\bn)}\mu^{kl}(\bn) $ and $ g(\bn)=
e^{h(\bn)}f(\bn).$

Clearly, the mask-object pair $(\nu,g)$ produces the identical set of
diffraction patterns as $(\mu, f)$. 
Therefore this ptychographic scheme has
at least $(2\tau_1-m)^2$ or $(2n-2\tau_1-m)^2$ degrees of ambiguity dimension
depending on whether $2\tau_1<n$ or $2\tau_1>n$. 

\end{ex}

The above construction of the ambiguity function $h$ extends to a perturbed scan \eqref{perturb2}  with any $q\ge 2$ and overlap ratios less than $50\%$. 
More importantly, $\tilde R_{00}$ and $\tilde R_{10}$ illustrate  the notion of {\em singly covered invariant regions} which may be present in more general schemes of  low overlap ratio. 

A singly covered  invariant region $R$ is the union of congruent subsets $R_j\subset\cM^j$ each of which is  
covered {\em once only} by the same subset $S\subset \cM^0$ of the mask, i.e.
$R_j= \bt_j+S$ for all $j$.   As in Example \ref{ex50}, the existence of such an invariant region entails an ambiguity function $h$ that is  any function defined on $\cM$ such that $h(\bn)=0$ for any $\bn\not\in R$ and $h(\bn+\bt_j)=h(\bn)$ for all $\bn$. In other words, every component region $R_j$ is infected with the same ambiguity which is transported by the mask region $S$ from component to component. 
The ambiguity dimension equals the size of each component region $R_j$.

In what follows, we further develop the ideas in Section \ref{sec:simple} and Example \ref{ex50} and formulate uniqueness conditions for more general shifts than the perturbed raster scan \eqref{perturb}. 
For simplicity of presentation, we focus on 3-part coupling which is most
relevant in the case of perturbed raster scans.

To this end, we resort to the  single-indexed notation in Section \ref{sec:two}. 

%Without loss of generality, we assume that $\bt_0=(0,0)$ and that $\{f^1,\cdots, f^k\}$ is the neighborhood of $f^0$, consisting of  all the adjacent nodes in $\Gamma_s$. Suppose $\sum_{i=1}^k p_i\bt_i=\ba$ for some $p_i\in \IZ$ and $\ba\in \IZ^2$. 
For   two neighbors of  $f^k$, say $f^{k-1}$ and $f^{k+1}$,  suppose 
\beq
\label{400.13}
p_1(\bt_{k}-\bt_{k-1})-p_2(\bt_{k+1}-\bt_k)=\ba
\eeq
 for some $p_1, p_2\in \IN$ and $\ba\in \IZ^2$. For ease of notation, set
 \[
 \bs_1=\bt_{k}-\bt_{k-1},\quad \bs_2= \bt_{k+1}-\bt_k.
 \]
 The same analysis is applicable to 
the other case $p_1\bs_1+ p_2\bs_2=\ba$.  %Without loss of generality, we assume that $\bt_0=(0,0)$. 

There are several paths for reducing   $h(\bn+p_1\bs_1-p_2\bs_2)$ to $h(\bn)$. 
Motivated by the example of perturbed raster scan, we can represent a path
of reduction from $p_1\bs_1-p_2\bs_2$ to $0$ by a directed path
on the $\IZ^2$-lattice spanned by $\bs_1$ and $\bs_2$ as in Figure \ref{fig3} (for $p_1=2, p_2=1$). 
Figure \ref{fig3} depicts three shortest (in the Manhattan metric) paths 
\beq
\sigma_1:&& 2\bs_2-\bs_2\longrightarrow 2\bs_1\longrightarrow\bs_1\longrightarrow 0\label{810}\\
\sigma_2:&&2\bs_2-\bs_2\longrightarrow \bs_1-\bs_2\longrightarrow\bs_1\longrightarrow 0\label{820}\\
\sigma_3:&&2\bs_2-\bs_2\longrightarrow \bs_1-\bs_2\longrightarrow -\bs_2\longrightarrow 0. \label{830}
\eeq

Let $\Pi(p_1,- p_2,\bs_1,\bs_2)$ denote the set of shortest paths (in the Manhattan metric) from $(p_1,-p_2)$ to $0$ in the
lattice spanned by $\bs_1$ and $\bs_2$. 

Each path $\sigma\in \Pi(p_1,-p_2,\bs_1,\bs_2)$ gives rise to an identity  
\beq\label{r20}
h(\bn+\ba)=h(\bn+p_1\bs_1-p_2\bs_2)=h(\bn)-\im p_2(\theta_{k+1}-\theta_k)+\im p_1(\theta_k-\theta_{k-1}) 
\eeq
(modulo $\im 2\pi$)  for all $\bn$ in the set
\beq\label{840}
(\cM^{k}-\ba) \cap D_k(\sigma,\bs_1, \bs_2),\quad D_k(\sigma, \bs_1, \bs_2):= \bigcap_{(u,v) \in \sigma} (\cM^k-u\bs_1-v\bs_2)
\eeq
where $(u,v)\in \sigma$ means all the grid points in the path $\sigma$,  excluding the two end points.

%The set of validity $\Sigma_0(\sigma,\bs_1,\bs_2)$ along  the path $\sigma$ is the intersection of the shifted blocks that are $\cM^0 $ minus the vectors  corresponding to the grid points on the {\em left or lower} end of any edge on the path. 

  \begin{figure}
\centering
\includegraphics[width=15cm,height=3cm]{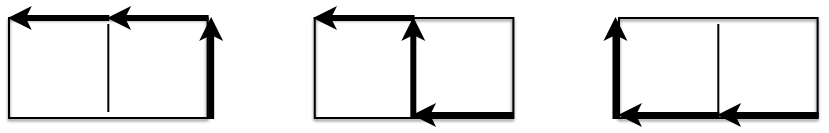}
\caption{Three shortest paths connecting $\ba$ to the origin (the upper-left corner) for $p_1=2, p_2=1$. The left diagram corresponds to $\sigma_1$ in \eqref{810}, the middle diagram to $\sigma_2$ in \eqref{820} and the right diagram to $\sigma_3$ in \eqref{830}.}  %Many other paths  (not shown) are longer and  not constrained within the  grid  most of which give rise to an empty $\Sigma_0$ and contribute nothing to \eqref{D55}.  }
\label{fig3}
\end{figure}

By repeatedly applying \eqref{200.5'} we can extend \eqref{r20} to a larger region as follows.
\begin{lem}\label{lem8.1}
The relation  \eqref{r20} holds 
 \beq
 \label{57}\label{200.112}
 h(\bn+\ba)&=&h(\bn)-\im p_2(\theta_{k+1}-\theta_k)+\im p_1(\theta_k-\theta_{k-1}),\quad\ba=p_1\bs_1-p_2\bs_2
 \eeq
 (modulo $\im 2\pi$)  holds true in the set
\beq
 \label{D55}\label{domain}
%\label{59}
\bigcup_{\bt\in \cT}\bigcup_{\sigma\in \Pi(p_1,-p_2,\bs_1,\bs_2)}\lt[\bt-\bt_k+D_k(\sigma,\bs_1,\bs_2) \cap (\cM^k-\ba)\cap\cM^k \rt]. 
\eeq
\end{lem}
%\commentout{%08/23/2019
\begin{proof}
For any fixed $\sigma$, we know from the above analysis that \eqref{57} holds true for all $\bn$ in the set \eqref{840}. 

By \eqref{200.5'},\beqn
h(\bn+\bt_l-\bt_k)&=&h(\bn)+\im\theta_l-\im\theta_k,\quad\forall\bn\in \cM^k,
\eeqn
and by \eqref{r20}
\beqn
 h(\bn+\ba+\bt_l-\bt_k)&=& h(\bn+\ba)+\im\theta_l-\im\theta_k\\
 &=& h(\bn)-\im p_2(\theta_{k+1}-\theta_k)+\im p_1(\theta_k-\theta_{k-1}) 
+\im\theta_l-\im\theta_k. 
\eeqn
Hence we have 
\beqn
h(\bn+\ba+\bt_l-\bt_k)&=& h(\bn+\bt_l-\bt_k)-\im p_2(\theta_{k+1}-\theta_k)+\im p_1(\theta_k-\theta_{k-1}).  
\eeqn
In other words, \eqref{57} is valid in the set $\bt_l-\bt_k+ \cM^k\cap (\cM^{k}-\ba) \cap D_k(\sigma,\bs_1,\bs_2)$. 
Taking the union over all shifts and paths, we obtain \eqref{D55}. 
\end{proof}
%}

 \commentout{
 We collect the preceding analysis into a lemma. 
\begin{lemma} \label{lem1} Suppose $\supp(f)=\cM$. Suppose  
\beq
\label{400.13}
p_1(\bt_{\ell_1}-\bt_{\ell_0})-p_2(\bt_{\ell_2}-\bt_{\ell_0})=\ba\in \IZ^2
\eeq
 for some $p_1, p_2\in \IZ$, $\ba\in \IZ^2$ and triplet 
$\bt_{\ell_0}, \bt_{\ell_1}, \bt_{\ell_2} \in \cT$.
Then 
\beq
%\label{200.10}
\label{200.112}
h(\bn+\ba)=h(\bn)+\im p_1(\theta_{\ell_1}-\theta_{\ell_0})-\im p_2(\theta_{\ell_2}- \theta_{\ell_0}) \mod \im 2\pi
\eeq
for all $\bn$ in the region 
\beq
 \label{D55}\label{domain}
\lefteqn{ D_{\ell_0}(\ba, \bt_{\ell_1}-\bt_{\ell_0}, \bt_{\ell_2}-\bt_{\ell_0},p_1,p_2)}\\
&:=&\bigcup_{\sigma\in \Pi(p_1,-p_2)}\bigcup_{\bt\in \cT} (\bt-\bt_{\ell_0}+\cM^{\ell_0}\cap(\cM^{\ell_0}-\ba)\cap \Sigma_{\ell_0}(\sigma,\bt_{\ell_1}-\bt_{\ell_0}, \bt_{\ell_2}-\bt_{\ell_0})) \nn
 \eeq 
 where $\Sigma_{\ell_0}(\sigma,\bt_{\ell_1}-\bt_{\ell_0}, \bt_{\ell_2}-\bt_{\ell_0})$ is the set of validity associated with the path $\sigma$. 
\end{lemma}
}

We now define the {\em mixing } schemes that connect different parts of the object by the ptychographic shifts
in a non-degenerate manner.

{\bf The Mixing Property}. {\em 
\commentout{Let
\[
 D_k(\sigma, \bs_1, \bs_2):= \bigcap_{(u,v) \in \sigma} (\cM^k-u\bs_1-v\bs_2)
\]
for a given path $\sigma$ in the integer-lattice spanned by $\bs_1$ and $\bs_2$.
Let $\Pi(p,-q,\bs_1,\bs_2)$ be the set of the shortest paths connecting $(p,-q)$ to (0,0) in the Manhattan metric. 
}
Let $\{(j^s_i, k^s_i, l^s_i)\}, s=1,2,$ be a non-empty subset of triplets of index such that
for some $p^s_i, q^s_i\in \IZ$
\beq
\label{D22}
&&\IZ^2_n\subseteq \bigcup_{\bt\in \cT}\bigcup_{\sigma}\lt[\bt-\bt_{k^s_i}+D_{i}(\sigma,\bt_{k^s_i}-\bt_{j^s_i},\bt_{l^s_i}-\bt_{k^s_i}) \cap (\cM^{k^s_i}-\ba^s_i)\cap\cM^{k^s_i} \rt]
\eeq
where $\sigma\in\Pi(p^s_i,-q^s_i,\bt_{k^s_i}-\bt_{j^s_i},\bt_{l^s_i}-\bt_{k^s_i})$ and 
$\ba^s_i:= p^s_i(\bt_{k^s_i}-\bt_{j^s_i})-q^s_i(\bt_{l^s_i}-\bt_{k^s_i}).$

Moreover, for some $c^s_i\in \IZ$ 
\beq
\label{mix}
\sum_{i} c^1_i \ba_i^1
=\bu_1,\quad \sum_{i} c^2_i \ba_i^2
=\bu_2
\eeq
 where $\{\bu_1,\bu_2\}$ is a $\IZ^2$-lattice basis. 
}

As seen in Theorems \ref{thm7.4} and \ref{thm7.5}, the most tedious part of the above definition is \eqref{D22} when
the set $D_{i}(\sigma,\bt_{k^s_i}-\bt_{j^s_i},\bt_{l^s_i}-\bt_{k^s_i}) \cap (\cM^{k^s_i}-\ba^s_i)\cap\cM^{k^s_i}$ is not rectangular.
%The definition is designed for generalization to irregular scanning schemes practiced in applications  \cite{circle, optimal, scan1}. The purpose behind the complex and daunting look of the definition is to extend the 3-part interactions to multi-part interactions in a two-tier approach. 

The mixing schemes are so named because the propagation of ambiguity by the ptychographic shifts, according to the phase drift equation \eqref{200.4}-\eqref{200.5'}, is so complete that
a distinct ambiguity profile (affine phase + scaling factor) emerges as a result. 

We can state the global uniqueness theorem for the mixing schemes whose proof is entirely analogous to that
of Theorem \ref{thm7.5}. 
\begin{theorem}\label{thm:mix} Suppose $\supp(f)=\IZ_n^2$.
If $\cT$ satisfies the mixing property, then 
\beq\label{101'}
h(\bn)&=&h(0)+\im \bn\cdot (r_1,r_2)\mod \im 2\pi, \\
\label{472''}
\phi(\bn)&=&\theta_0-\Im[h(0)]-\bn\cdot(r_1,r_2) \mod 2\pi\\
\alpha&=&e^{-\Re[h(0)]}\label{473''}\\
\label{200.117}
\theta_\bt&=&\theta_0
+\bt\cdot(r_1,r_2)\mod 2\pi,\quad\forall \bt\in \cT,
\eeq
for some $ r_1, r_2\in \IR$ and all $\bn\in \IZ_n^2$ where $\alpha$ is given in \eqref{2mask}.

\commentout{then $h(\bn)$ is an affine function
\beq\label{300.1}
h(\bn)=h(0)+\im \bn\cdot\br \mod \im 2\pi,\quad\bn\in \cM,
\eeq
with
\beq
\label{300.2}
\br=\lt[\begin{matrix}
\sum_{i=1}^2  b_{1i}\sum_{(\ell_0,\ell_1,\ell_2)} c^i\lt[p^i_1
(\theta_{\ell_1}-\theta_{\ell_0})-p^i_2 (\theta_{\ell_2}-\theta_{\ell_0})\rt]\\
\sum_{i=1}^2  b_{2i}\sum_{(\ell_0,\ell_1,\ell_2)} c^i\lt[p^i_1
(\theta_{\ell_1}-\theta_{\ell_0})-p^i_2 (\theta_{\ell_2}-\theta_{\ell_0})\rt]
\end{matrix}\rt].
\eeq
 and hence by \eqref{200.1}
\beq
\label{200.116}\label{phase-error}
\alpha(\bn)=e^{-\Re[h(0)]},\quad \phi(\bn)&=&\phi(0)-\bn\cdot\br \mod 2\pi,\quad \forall \bn\in \cM^0.
\eeq
}
\end{theorem}

\section{Conclusion and discussion}\label{sec:con}

\begin{figure}
\centering
\subfigure[Random independent mask]{\includegraphics[width=5.4cm,height=5cm]{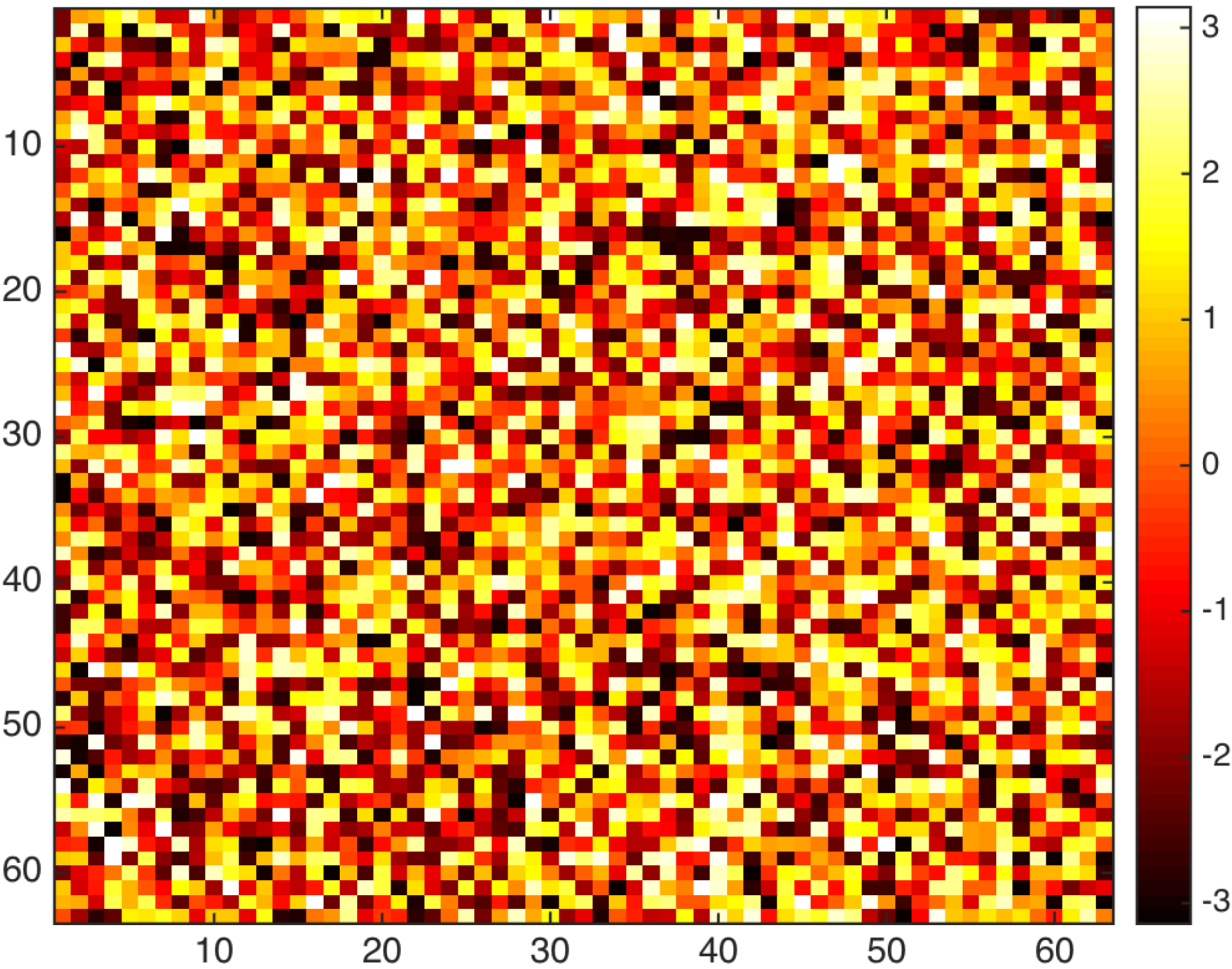}}\hspace{2cm}
%\subfigure[Correlated probe $c=0.4$]{\includegraphics[width=7cm]{figs/correlated-04.pdf}}\\
\subfigure[Correlation length = 0.7 $\times$ mask size]{\includegraphics[width=5.4cm,height=5cm]{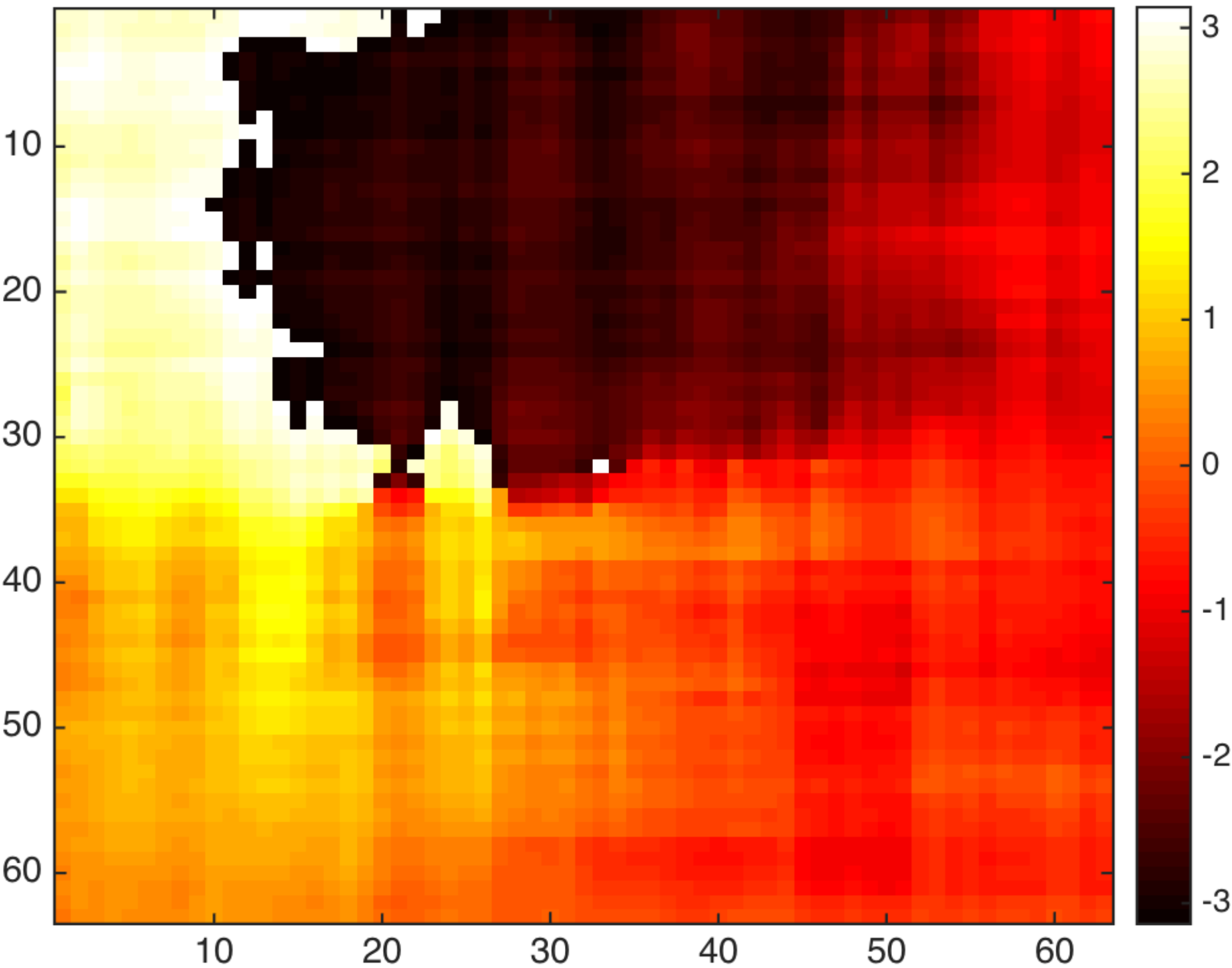}}\quad
%\subfigure[Correlated probe $c=1$]{\includegraphics[width=7cm]{figs/correlated-1.pdf}}
\caption{The phase profile of (a) the random independent mask  and (b) the correlated mask  of correlation length equal to
0.7 mask size.}
\label{fig:probes}
\end{figure}

Under the Mask Phase Constraint (MPC) and the  anchoring assumption, we have proved, for a strongly connected object, the local uniqueness (Theorem \ref{thm:u} and Theorem \ref{thm:many}) manifested as the phase drift equation \eqref{200.4}-\eqref{200.5'}.  We have shown by examples (Examples \ref{ex0} and \ref{ex3.1}) that
both MPC and the anchoring assumption are necessary. 
For the global uniqueness with the exception of inherent ambiguities (scaling factor and affine phase factor),
we have showed that  the mixing schemes  are ptychographically complete (Theorem \ref{thm:mix}), including
the perturbed raster scans (Theorems \ref{thm7.4} and \ref{thm7.5}). 

In addition, for both the mixing schemes and the regular raster scan (Proposition \ref{prop:raster}), we have proved that
their block phases must have an affine profile, $\theta_\bt=\theta_0+\bt\cdot \br$ for some $\br\in \IR^2$. It is unclear
if this holds true for any other schemes without the global uniqueness property.

Our approach to global uniqueness  is based on 3-part coupling designed particularly for analyzing the perturbed raster scans.
Our  theory and Example \ref{ex50} prove that
 the overlap ratio $50\%$ is more or less the minimum requirement for blind ptychography with the irregularly perturbed raster scan (see \eqref{coprime0}, \eqref{805}).

Our theory has several practical implications. First,  the connectivity condition \eqref{r100} suggests that  in the case of a sparse object a higher overlap ratio may be required.  Second,  
 MPC is re-interpretable  in terms of other measurement uncertainties such as scan position errors \cite{Fie08}. 
 The level of scan position off-sets that can be corrected depends on the type of mask used in measurement. 
For a random independent mask (Figure \ref{fig:probes}(a)), MPC corresponds to  correctable position error of about half a pixel;
For a correlated mask, MPC corresponds to correctable position error on the order of the correlation length.

In other words, there is  a trade-off between the mask correlation length and the correctable  level of scan position error. 
Numerical evidence suggests that with the same MPC a highly correlated  mask (Figure \ref{fig:probes}(b)) performs only slightly worse than the random independent mask \cite{DRS-ptych}. As mechanical and thermal vibrations are inevitable, it makes sense to use a mask of a comparable 
correlation length to compensate for scan position offset. On the other hand, a simple regular mask (e.g. Fresnel illumination spot) is often a sub-optimal  choice
as twin-like ambiguities may be present even with perfect knowledge of the mask \cite{ptych-unique}. 
In addition, a random mask has the benefit of producing more diffuse illumination and thus data of lower dynamic-range.

Another implication of MPC is in numerical reconstruction.  MPC is independent of the knowledge about the mask amplitude, meaning that
the knowledge about the mask phase is much more important for blind ptychography. Indeed, 
 MPC turns out to be  an effective method for mask
initialization, yielding geometrically convergent iterations, even when the initialization error (measured in L2 norm) is  large  \cite{DRS-ptych}. 

\begin{figure}
\centering
\subfigure[concentric circles]{
\includegraphics[width=5cm,height=5cm]{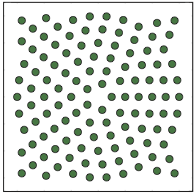}}\hspace{2cm}
\subfigure[Fermat spiral]{
\includegraphics[width=5cm,height=5cm]{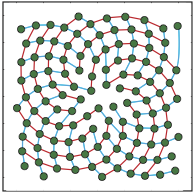}}
\caption{Special scans that have good empirical performances \cite{circle, optimal, scan1}. }
\label{fig4}
\end{figure}

One  is naturally  led to the important question of optimal scan schemes  which  use  the minimum number of diffraction patterns for a given object (i.e. the minimum redundancy in measurement) to be ptychographically complete. The measurement redundancy is more or less proportional to  the product of the number of adjacent blocks and the overlap ratio.   
Our results show that the irregularly perturbed raster scans with overlap ratio slightly over $50\%$ is optimal among
the class of perturbed raster scans. 
For more general scans, the minimum overlap requirement may be lowered and 5-part or higher order coupling must be directly accounted for. For example, the Fermat spiral scan scheme (Figure \ref{fig4}(b)) is claimed to provide a more uniform coverage than the perturbed raster scans and
the concentric circle pattern (Figure \ref{fig4}(a)), thus lowering  the  overlap ratio  \cite{scan1}.  A rigorous theory for general optimal scans, however,  is beyond the scope of the present work and has to be left for future research.

 \appendix
 
 \section{Object support constraint (OSC)}\label{sec:osc}
 Instead of a tight support, an object part may possess various degrees of loose support depending on the scan position
 and size of the block (see Figure \ref{fig:corn}). The looseness of support can be  characterized by a set of admissible shifts $T_0$ as follows. 
  
 {\bf Object Support Constraint (OSC):}
{\em  An object estimate $g^0$ satisfies the Object Support Constraint (OSC) with respect to a given set of shifts $
T_0
$ 
 if  $\mbm\in T_0$ whenever
 \beq
 \label{shifts2}
 %\supp(g^0)\mb{\rm\quad or \quad } \supp(\mb{\rm Twin}(g^0))\subseteq\mb{\rm Box}[\supp(f^0)]-T_0,\\
\commentout{%09/21/2018
 && \supp(g^0)\cup   \supp(\mb{\rm Twin}(g^0))\subseteq\mb{\rm Box}[\supp(f^0)]-T_0
 }
 \supp(g^0) \quad\mb{\rm or}\quad  \supp(\mb{\rm Twin}(g^0)) \subseteq\mb{\rm Box}[\supp(f^0)]-\mbm.
%\label{shifts2'}   \supp(\mb{\rm Twin}(g^0))\subseteq\mb{\rm Box}[\supp(f^0)]-T_0 %\label{shifts2'} \supp(g^\bt)\mb{\rm \quad or \quad} \supp(\mb{\rm Twin}(g^\bt))\subseteq\mb{\rm Box}[\supp(f^\bt)]-\cS',
%\label{shifts2'}\mb{\rm or\quad}&& \supp(g^\bt)\bigcup \supp(\mb{\rm Twin}(g^\bt))\subseteq\mb{\rm Box}[\supp(f^\bt)]-T_1
\eeq
}

We can use OSC to describe the precision of our prior knowledge about $ \mb{\rm Box}[\supp(f^0)]$ when $f^0$ has  a loose support in $\cM^0$. The smaller the set $T_0$ is, the more precise the OSC is. When  $\mb{\rm Box}[\supp(f^0)]=\cM^0$,  we can set $T_0=\{(0,0)\}$ since the condition \eqref{shifts2} becomes
\beq
\label{osc2}
\commentout{%09/21/2018
\supp(g^0)\cup   \supp(\mb{\rm Twin}(g^0))\subseteq \cM^0
}
 \supp(g^0) \quad\mb{\rm or}\quad  \supp(\mb{\rm Twin}(g^0))\subseteq \cM^0
\eeq
which is null and gives  no new information. 

Under OSC, the quantity $s$ in \eqref{3.2.1} is defined instead as
  \beq
 \label{3.2.1'}
s=\min_{\mbm,\mbm'\in T_0}|S_0(\mbm)|\wedge |S'_0(\mbm')|\ge 2
\eeq
where $T_0$ is the set of shifts in OSC and 
 \beq
S_{0}(\mbm)&=&\cM^0\cap \cM^\bt\cap (\supp(f^0)-\mbm)\nn\\
%S_1(\mbm)&=& \cM^0\cap \cM^\bt\cap [\supp(f^\bt)-\mbm)]\nn\\
S'_{0}(\mbm)&=&\cM^0\cap \cM^\bt\cap (\supp( \mb{\rm Twin}(f^0))+\mbm).\nn
%S'_{1}(\mbm)&=&\cM^0\cap \cM^\bt\cap [\supp(\mb{\rm Twin}(f^\bt))+\mbm].\nn
\eeq

The construction in Example \ref{ex3.1} satisfies the OSC \eqref{shifts2} with 
\[
T_0=\Big\{(a,0): a=0,\dots, m/2\Big\}.%\quad \cS'= \Big\{(a,0): a=-m/3,\dots, 0\Big\}  
\]
On the other hand, if  $f^{0}_{1}, f^{1}_{0}$ are non-vanishing, then it can be verified 
that 
$s=0$, consistent with
the fact that the probability for ambiguity is one as shown in the above construction. 

However, if we enhance the precision of the support knowledge by tightening  $T_0$ by any amount $l\ge 1$ as 
\beq
\label{3.30}
T_0= \Big\{(a,0): a=0,\dots, m/2-l\Big\},%\quad \cS'=\Big\{(a,0): a=-m/3+l,\dots, 0\Big\},\quad l\ge 1
\eeq
then the  constructions  would violate the OSC \eqref{shifts2}, and be rejected. Moreover, for \eqref{3.30}, $s=ml$ with nonvanishing $f^{0}_{1}, f^{1}_{0}$ so  the probability of uniqueness is closed to one  for
$m\gg 1$ as predicted by Theorem \ref{thm:u}. 

Although OSC is more general than the anchoring assumption, it is also more complicated and less practical  so we do not
pursue the full proof here. For the interested reader, we refer to
 the preliminary version \cite{blind-ptycho} for the proof of Theorem \ref{thm:u} under the assumption of OSC.

 \section{Proof of Theorem \ref{thm:u}}\label{sec:proof}
%\begin{proof} (Theorem \ref{thm:u})

Let  $\bN=(m,m)$. Applying  Corollary \ref{cor1} to both $\cM^0$ and $\cM^\bt$
we have the following alternatives: For  some $\mbm_1,\mbm_2\in \IZ^2, \theta_0,\theta_\bt\in \IR$.
\beq
\label{1.60}
g^0(\bn)&=& e^{\im \theta_0} f^0(\bn + \mbm_1)\mu^0(\bn + \mbm_1)/\nu^0(\bn) \\
\nn\mb{\rm or\quad Twin}(g^0)(\bn)&=& e^{-\im \theta_0} f^0(\bn+\mbm_1) \mu^0(\bn +\mbm_1)/\mb{\rm Twin}(\nu^0)(\bn),\quad\forall \bn \in \cM^0
\eeq
and
\beq
\label{1.70}
g^\bt(\bn)&=&e^{\im \theta_\bt} f^\bt(\bn + \mbm_2)\mu^\bt(\bn + \mbm_2)/\nu^\bt(\bn) \\
\nn\mb{\rm or\quad Twin}(g^\bt)(\bn)&=& e^{-\im \theta_\bt} f^\bt(\bn+\mbm_2) \mu^\bt(\bn +\mbm_2)/\mb{\rm Twin}(\nu^\bt)(\bn),\quad\forall \bn \in \cM^\bt.
\eeq
Note that $\mbox{\rm Twin}(g^\bt)(\bn)=\bar g^\bt(\bN+2\bt-\bn)$ so  we can rewrite \eqref{1.60} and
\eqref{1.70} as
\beq
\label{1.6}
g^0(\bn)&=& e^{\im \theta_0} f^0(\bn + \mbm_1)\mu^0(\bn + \mbm_1)/\nu^0(\bn) \\
\nn&\mbox{\rm or}& e^{\im \theta_0} \bar f^0(\bN-\bn+\mbm_1)\bar \mu^0(\bN-\bn +\mbm_1)/\nu^0(\bn),\quad\forall \bn \in \cM^0
\eeq
and
\beq
\label{1.7}
g^\bt(\bn)&=&e^{\im \theta_\bt} f^\bt(\bn + \mbm_2)\mu^0(\bn + \mbm_2-\bt)/\nu^0(\bn-\bt) \\
\nn&\mbox{\rm or}& e^{\im \theta_\bt} \bar f^\bt(\bN+2\bt-\bn +\mbm_2)\bar\mu^0(\bN+\bt-\bn +\mbm_2)/\nu^0(\bn-\bt),\quad\forall \bn \in \cM^\bt
\eeq
for  some $\mbm_1,\mbm_2\in \IZ, \theta_0,\theta_\bt\in \IR$ where
we have used the relation $\mu^\bt(\cdot)=\mu^0(\cdot-\bt),\nu^\bt(\cdot)=\nu^0(\cdot-\bt)$. Note that
$\bN$ and $\bN+\bt=(m+t_1, m+t_2)$  are the upper-right corners of $\cM^0$ and $\cM^\bt$, respectively. 
\commentout{
\beq
\label{1.6}
g^0(\bn)&=& e^{\im \theta_0} f^0(\bn + \mbm_1)\mu^0(\bn + \mbm_1)/\nu^0(\bn) \\
\nn\mb{\rm or\quad Twin}(g^0)(\bn)&=& e^{-\im \theta_0} f^0(\bn+\mbm_1) \mu^0(\bn +\mbm_1)/\mb{\rm Twin}(\nu^0)(\bn),\quad\forall \bn \in \cM^0
\eeq
and
\beq
\label{1.7}
g^\bt(\bn)&=&e^{\im \theta_\bt} f^\bt(\bn + \mbm_2)\mu^\bt(\bn + \mbm_2)/\nu^\bt(\bn) \\
\nn\mb{\rm or\quad Twin}(g^\bt)(\bn)&=& e^{-\im \theta_\bt} f^\bt(\bn+\mbm_2) \mu^\bt(\bn +\mbm_2)/\mb{\rm Twin}(\nu^\bt)(\bn),\quad\forall \bn \in \cM^\bt
\eeq
for some $\mbm_1,\mbm_2, \theta_0,\theta_\bt$. 
}

In view of  the anchoring assumption, \eqref{1.60} implies 
$\mbm_1=0$.

We now focus on the intersection $\cM^0\cap \cM^\bt$ where \eqref{1.6} and \eqref{1.7}
both hold. We have then four possible ambiguities from the crossover of the alternatives in \eqref{1.6} and \eqref{1.7}.  

%For ease of notation, let us denote
%\[
%S_0=\cM^0\cap \cM^\bt\cap \supp(f^0), \quad S_0'=\cM^0\cap \cM^\bt\cap \supp( \mb{\rm Twin}(f^0)).
%\]

{\bf Case (i).} The combination of the first alternatives in \eqref{1.6} and \eqref{1.7} imply that for all $\bn\in \cM^0\cap \cM^\bt$
\beq
\label{1.9}
{e^{\im\theta_0}f^0(\bn )\mu^0(\bn)/\nu^0(\bn)}
&=&e^{\im \theta_\bt}  f^\bt(\bn + \mbm_2 )\mu^0(\bn-\bt + \mbm_2)/\nu^0(\bn-\bt)
\eeq
provided that $f^0(\bn)$ and $f^\bt(\bn+\mbm_2)$ are both zero or nonzero. 

We now show that with high probability \eqref{1.9} fails to hold  for some  $\bn\in \cM^0\cap \cM^\bt$.  

Consider any $ \bn\in S_0$ (hence  
$f^0(\bn)\neq 0$) and assume that $f^\bt(\bn+\mbm_2)\neq 0$. Otherwise, 
\eqref{1.9} holds with probability zero. 

We obtain by taking logarithm on both sides of \eqref{1.9} that  \beq
\label{1.10}
&&\ln{\mu^0(\bn)}+\ln{\mu^0(\bn-\bt)}-\ln{\mu^0(\bn-\bt+\mbm_2)}-\ln{\mu^0(\bn)}\\
&=&\im\theta_\bt-\im\theta_0  -\ln{f^0(\bn)}+\ln{f^\bt(\bn+\mbm_2)}+\ln \alpha(\bn)-\ln \alpha(\bn-\bt)\nn\\
&&+\im\phi(\bn)-\im\phi(\bn-\bt)\nn
\eeq
modulo  $\im 2\pi$. We want to show that if $\lt|S_0\rt|$ is sufficiently large
then \eqref{1.10} holds with at most exponentially small probability.

Since $\bn\in \cM^0\cap\cM^\bt$ and $\bn+\mbm_2\in \cM^\bt$,  the points associated with the lefthand side of \eqref{1.10}, $ \bn-\bt, \bn+\mbm_2-\bt$,  belong in $ \cM^0$. 
Hence the random variables on the lefthand side of \eqref{1.10} are well-defined and have a finite value. 

The two points $ \bn-\bt, \bn+\mbm_2-\bt$  can not be identical unless
$\mbm_2=0$. In other words,  if $\mbm_2\neq 0$, then
the imaginary part $\Theta_1$ of the  lefthand side of \eqref{1.10} 
\beq
\label{r3}
\Theta_1:=\theta(\bn-\bt)- \theta(\bn-\bt+\mbm_2)
\eeq
is the sum of two independent random variables and hence the support set of its probability density contains $(-2\gamma\pi,2\gamma\pi]$. 

On the righthand side of \eqref{1.10}, however, as $f^0(\bn)$ and $f^\bt(\bn+\mbm_2)$ are given (hence deterministic), the phase fluctuation is determined by $\phi(\bn)-\phi(\bn-\bt)$ which ranges over the interval $(-2\delta\pi,2\delta\pi]$ due to
the constraint \eqref{u1}. Consequently  \eqref{1.10} holds true with
probability at most 
\beqn
p_1:=\max_{a\in \IR}\mbox{\rm Pr} \{\Theta_1\in (a-2\delta\pi, a+2\delta\pi]\}<1,
\eeqn
for each $\bn$, since $\delta<\min(\gamma,\half)$. 

For all $\bn \in S_0$,  there are at least $|S_0|/2!$ statistically independent  instances, corresponding to the number of  {\em non-intersecting} $
 \{\bn-\bt, \bn+\mbm_2-\bt\}$. 
 Therefore \eqref{1.10} holds true 
 with probability at most $p_1^{|S_0|/2!}$ unless $\mbm_2=0$. 

On the other hand, for $\mbm_2=0$, the desired result \eqref{masked1}-\eqref{masked2} follows directly from the first alternatives in \eqref{1.6} and \eqref{1.7}. 

{\bf Case (ii).}  Consider the combination of the first alternative in \eqref{1.6} and the second alternative in \eqref{1.7} that for $\bn\in \cM^0\bigcap\cM^\bt$ 
\beq\label{1.12'}
g(\bn)&=&e^{\im \theta_0} f^0(\bn )\mu^0(\bn)/\nu^0(\bn)\\
&=&e^{\im \theta_\bt} \bar f^\bt(\bN+2\bt-\bn +\mbm_2)\bar \mu^0 (\bN+\bt-\bn +\mbm_2)/\nu^0(\bn-\bt),\nn
\eeq
provided that $f^0(\bn)$ and $\bar f^\bt(\bN+2\bt-\bn +\mbm_2 )$ are  
both zero or nonzero.

Consider  any $ \bn\in S_0$ (hence  
$f^0(\bn)\neq 0$) and  assume $\bar f^\bt(\bN+2\bt-\bn+\mbm_2)\neq 0$. Otherwise \eqref{1.12'} is false and can be ruled out.

Taking logarithm and rearranging terms in \eqref{1.12'} we have
\beq\label{1.12}
&&{\ln{\mu^0(\bn)}+\ln{\mu^0(\bn-\bt)}-\ln{\bar\mu^0(\bN+\bt-\bn+\mbm_2)}}-  \ln{\mu^0(\bn)}\\
&=&\im\theta_\bt-\im\theta_0-\ln{f^0(\bn)}+\ln{\bar f^\bt(\bN+2\bt-\bn+\mbm_2)}+\ln \alpha(\bn)-\ln \alpha(\bn-\bt)\nn\\
&&+ \im \phi(\bn)-\im\phi(\bn-\bt).\nn
\eeq
The imaginary parts of the lefthand side of \eqref{1.12} 
\beq
\label{1.120}
\Theta_2:=\theta(\bn-\bt)+\theta(\bN+\bt-\bn+\mbm_2)
\eeq
is the sum of two independent random variables unless
\beqn
 \bn=\bt+\half(\bN+\mbm_2), 
\eeqn
in which case $\Theta_2=2\theta(\bn-\bt)$. Since $|S_0|\ge 2$,  there exists some $\bn\in S_0$ such that $\Theta_2$ is the sum of 
two independent random variables and hence the support of its probability density function
contains $(-2\gamma\pi,2\gamma\pi]$. 
By the same argument as above, \eqref{1.12} holds true  with  probability at most $p_2^{|S_0(\mbm_1)|/2!}$ where
\[
p_2:=\max_{a\in \IR}\mbox{\rm Pr} \{\Theta_2\in (a-2\delta\pi, a+2\delta\pi]\}<1 
\]
since $\delta<\min(\gamma,\half)$. 

{\bf Case (iii).}  Consider the combination of the second alternative in \eqref{1.6} and the first alternative in \eqref{1.7}  that for $\bn\in \cM^0\bigcap\cM^\bt$ 
\beq\label{r7'}
g(\bn)&=& e^{\im \theta_0} \bar f^0(\bN-\bn)\bar \mu^0 (\bN-\bn)/\nu^0(\bn)\\
&=&e^{\im \theta_\bt} f^\bt(\bn + \mbm_2)\mu^0(\bn -\bt +\mbm_2)/\nu^0(\bn-\bt)  \nn
\eeq
provided that $f^0(\bn)$ and $f^\bt(\bN+2\bt-\bn +\mbm_2 )$ are  
both zero or nonzero.
Consider any $ \bn\in S'_0$ (hence  
$ \bar f^0(\bN-\bn)\neq 0$) and assume $f^\bt(\bn+\mbm_2)\neq 0$. Otherwise \eqref{r7} can be ruled out.

Taking logarithm and rearranging terms in \eqref{r7'} we have
\beq
\label{r7}
&&{\ln{\bar\mu^0(\bN- \bn)}+\ln{\mu^0(\bn-\bt)}-\ln{\mu^0(\bn-\bt+\mbm_2)}}-  \ln{\mu^0(\bn)}\\
&=&\im\theta_\bt-\im\theta_0-\ln{\bar f^0(\bN-\bn)}+\ln{f^\bt(\bn+\mbm_2)}+\ln \alpha(\bn)-\ln \alpha(\bn-\bt)\nn\\
&&+ \im \phi(\bn)-\im\phi(\bn-\bt).\nn
\eeq
As before, we want to show that if $|S_0'|$ is sufficiently large, then  \eqref{r7} holds with at most exponentially small probability. 

Since $\bn\in \cM^0\cap \cM^\bt$ and $\bn+\mbm_2\in \cM^\bt,$ the four points associated with the lefthand side of \eqref{r7}, $\bN-\bn, \bn-\bt,\bn-\bt+\mbm_2, \bn$, belong in $\cM^0$. Hence
the four random variables on the lefthand side of \eqref{r7} are well-defined. 

 The imaginary parts of the lefthand side of \eqref{r7} given by 
\beq
\label{r8}
\Theta_3:=-\theta(\bN-\bn)+\theta(\bn-\bt)-\theta(\bn-\bt+\mbm_2)-\theta(\bn)
\eeq
is the sum of  two, three or four independent random variables unless
%The two sub-sets $\{\bN-\bn,\bn-\bt\}$ and $\{\bn-\bt+\mbm_2, \bn\}$ can not coincide unless
%\beqn
%\mbm_2=\bt,\quad \bn=\half(\bN+\bt)
%\eeqn
%or
\beq
\label{990}
\mbm_2=0,\quad \bn=\half\bN, 
\eeq
in which case $\Theta_3=2\theta(\bN/2)$. 
  
Since $S'_0\ge 2$,  there exists some $\bn\in S'_0$ such that $\Theta_2$ is the sum of at least
two independent random variables and hence the support of its probability density function
contains $(-2\gamma\pi,2\gamma\pi]$. 

On the righthand side of \eqref{r7}, the phase fluctuation is determined by $\phi(\bn)-\phi(\bn-\bt)$ which ranges over the interval $(-2\delta\pi, 2\delta\pi]$ due to the \mpc \eqref{u1}. 
So \eqref{r7} holds true  with  probability at most 
\[
p_3:=\max_{a\in \IR}\mbox{\rm Pr} \{\Theta_3\in (a-2\delta\pi, a+2\delta\pi]\}<1 
\]
for each $\bn$, since $\delta<\min(\gamma,\half)$. 

For all $\bn\in S_0'$ such that $\bn\neq \bN/2$, there are at least $(|S_0'|-1)/4!$ statistically independent instances, corresponding to the number of non-intersecting $\{\bN-\bn, \bn-\bt,\bn-\bt+\mbm_2, \bn\}$
Therefore, \eqref{r7} holds true with probability at most 
$p_3^{(|S'_0|-1)/4!}$.

{\bf Case (iv).}  Now consider the combination of the second alternatives in \eqref{1.6} and \eqref{1.7} that for $\bn\in \cM^0\bigcap\cM^\bt$ 
\beq
\label{r5}
g(\bn)&=&e^{\im \theta_0} \bar f^0(\bN-\bn)\bar \mu^0 (\bN-\bn)/\nu^0(\bn)\\
&=& e^{\im \theta_\bt} \bar f^\bt (\bN+2\bt-\bn +\mbm_2)\bar \mu^0 (\bN+\bt-\bn +\mbm_2)/\nu^0(\bn-\bt)\nn
\eeq
provided that $\bar f^0(\bN-\bn)$ and $\bar f^\bt(\bN+2\bt-\bn +\mbm_2 )$ are  
both zero or nonzero.

Consider any $ \bn\in S'_0$ (hence  
$ \bar f^0(\bN-\bn)\neq 0$) and assume $ \bar f^\bt(\bN+2\bt-\bn+\mbm_2)\neq 0$. Otherwise \eqref{r5} is ruled out.

 After taking logarithm and rearranging terms for $\bn\in S_0'$  \eqref{r5} becomes
 \beq
\label{1.13}
&&{\ln{\bar\mu^0(\bN-\bn)}+\ln{\mu^0(\bn-\bt)}-\ln{\bar\mu^0(\bN+\bt-\bn+\mbm_2)} -\ln{\mu^0(\bn)}}\\
&=&\im\theta_\bt-\im\theta_0- \ln{\bar f^0(\bN-\bn)}+\ln{\bar f^\bt(\bN+2\bt-\bn+\mbm_2)}\nn\\
&& +\ln \alpha(\bn)-\ln \alpha(\bn-\bt)+\im\phi(\bn)-\im\phi(\bn-\bt).\nn
\eeq
The imaginary part of the lefthand side of \eqref{1.13} 
\beq
\label{1.130}
\Theta_4:=-\theta(\bN-\bn)+\theta(\bn-\bt)+\theta(\bN+\bt-\bn+\mbm_2)-\theta(\bn)
\eeq
 is the sum of two, three or four independent random variables
unless 
\beqn
\bN+\bt-\bn+\mbm_2&=&\bn\\
\bN-\bn&=& \bn-\bt
\eeqn
or equivalently 
\beqn
\mbm_2=0, \quad \bn=\half (\bN+\bt). 
\eeqn
Since $|S'_0|\ge 2$, the support of the probability density of $\Theta_4$
contains  $(-2\gamma\pi,2\gamma\pi]$. 

The same analysis then implies that  \eqref{1.13} holds true 
with  probability at most $p_4^{(|S'_0|-1)/4!}$ where
\[
p_4:=\max_{a\in \IR}\mbox{\rm Pr} \{\Theta_4\in (a-2\delta\pi, a+2\delta\pi]\}<1 
\]
since $\delta<\min(\gamma,\half)$.

In summary, ambiguities (i)-(iv) are present with probability at most $c^s$ and 
hence the desired result \eqref{masked1}-\eqref{masked2} holds true with probability greater than $1-c^s$ where the positive constant $c<1$ depends only on $\delta$ and  the probability density function of the mask phase. 
%\end{proof}

\section{Proof of Theorem \ref{thm:many}}\label{sec:many}

Without loss of generality, we may assume $\ell_0=0$.

Let $\cM^{\ell(k)}$ denote an adjacent block of $\cM^k$ such that
$f^{\ell(k)}$ and $f^k$ are $s-$connected. When the $s$-connected neighbor of $\cM^k$ is not unique, we make an arbitrary selection  $\ell(k)$ such that $\ell(\ell(k))=k$. Let $L_j= \{f^k, f^{\ell(k)}: k=0,\dots, j\}$. 

We prove \eqref{3.100} 
by induction. Suppose that \eqref{3.100} holds for $k=0,\dots,j$. We wish to show that there is another  part, say $f^{j+1}\not\in L_j$, such that \eqref{3.100} holds for $k=0,\dots,j, j+1,$ unless $j=Q-1$. 
Since $\{f^k:k=0,\cdots, Q-1\}$ is $s$-connected,  at least some $f^{j+1}$ is $s$-connected to, say $f^l\in L_j$ if $j<Q-1$. 

Denote $S_0:=\cM^l\cap \cM^{j+1}\cap\supp(f)$.  %and $S_0':=\cM^l\cap \cM^{j+1}\cap \supp(\mb{\rm Twin}(f^l))$.
Applying  Corollary \ref{cor1} to $\cM^{j+1}$
we have the following alternatives: For  some $\mbm\in \IZ, \theta\in \IR$,
\beq
\label{1.70'}
g^{j+1}(\bn)&=&e^{\im \theta} f^{j+1}(\bn + \mbm)\mu^{j+1}(\bn + \mbm)/\nu^{j+1}(\bn) \\
\nn\mb{\rm or\quad Twin}(g^{j+1})(\bn)&=& e^{-\im \theta} f^{j+1}(\bn+\mbm) \mu^{j+1}(\bn +\mbm)/\mb{\rm Twin}(\nu^{j+1})(\bn),\quad\forall \bn \in \cM^{j+1}.
\eeq
Let $\cM^{j+1}=\cM^l+\bt$ for some shift $\bt$.

Consider the first alternative for $\bn\in \cM^l\cap \cM^{j+1}$:
\beq
\label{1.9'}
{e^{\im\theta_l}f^{l}(\bn)\mu^{l}(\bn)/\nu^{l}(\bn)}
&=&e^{\im \theta}  f^{j+1}(\bn + \mbm)\mu^{j+1}(\bn+ \mbm)/\nu^{j+1}(\bn)\\
&=&e^{\im \theta}  f^{j+1}(\bn + \mbm)\mu^{l}(\bn-\bt+ \mbm)/\nu^{l}(\bn-\bt)\nn
\eeq
provided that $f^l(\bn)$ and $f^{j+1}(\bn+\mbm)$ are both zero or nonzero. 

Suppose $ f^l(\bn)\cdot f^{j+1}(\bn+\mbm)\neq 0$. 
We obtain by taking logarithm on both sides of \eqref{1.9'} that  \beq
\label{1.10'}
&&\ln{\mu^l(\bn-\bt)}-\ln{\mu^l(\bn-\bt+\mbm)}\\
&=&\im\theta-\im\theta_l  -\ln{f^l(\bn)}+\ln{f^{j+1}(\bn+\mbm)}+\ln \alpha(\bn)-\ln \alpha(\bn-\bt)\nn+\im\phi(\bn)-\im\phi(\bn-\bt)\nn
\eeq
modulo  $\im 2\pi$. We want to show that if $s$ is sufficiently large
then \eqref{1.10'} holds with at most exponentially small probability unless $\mbm=0$.

Since $\bn\in \cM^l\cap\cM^{j+1}$ and $\bn+\mbm\in \cM^{j+1}$,  $ \bn-\bt$ and $\bn+\mbm-\bt$  belong in $ \cM^l$. 
Hence  the lefthand side of \eqref{1.10'} is well-defined and has a finite value. 

Unless $\mbm=0$, 
the imaginary part $\Theta_1$ of the  lefthand side of \eqref{1.10'} 
\beqn
\Theta_1:=\theta(\bn-\bt)- \theta(\bn-\bt+\mbm)
\eeqn
is the sum of two independent random variables and hence the support set of its probability density contains $(-2\gamma\pi,2\gamma\pi]$. 

On the righthand side of \eqref{1.10'}, however, as $f^l(\bn)$ and $f^{j+1}(\bn+\mbm)$ are deterministic, the phase fluctuation is determined by $\phi(\bn)-\phi(\bn-\bt)$ which is limited to the interval $(-2\delta\pi,2\delta\pi]$ due to
\mpc. Consequently  \eqref{1.10'} holds true with
probability at most 
\beqn
p_1:=\max_{a\in \IR}\mbox{\rm Pr} \{\Theta_1\in (a-2\delta\pi, a+2\delta\pi]\}<1,
\eeqn
for each $\bn$, since $ \delta<\min\{\gamma,1/2\}$.

For all $\bn \in S_0$,  there are at least $|S_0|/2$ statistically independent  instances, corresponding to the number of  {\em non-intersecting} $
 \{\bn-\bt, \bn+\mbm-\bt\}$. 
 Therefore \eqref{1.10'} holds true 
 with probability at most $p_1^{|S_0|/2}$ unless $\mbm=0$. 
On the other hand, for $\mbm=0$, the desired result \eqref{3.100} for $k=j+1$ follows directly from \eqref{1.70'}.

Consider the second alternative  in \eqref{1.70'} and note that
\[
\mbox{\rm Twin}(g^{j+1})(\bn)=\bar g^{j+1}(\bN+2\bt_{j+1}-\bn),\quad \mbox{\rm Twin}(\nu^{j+1})(\bn)=\bar \nu^{j+1}(\bN+2\bt_{j+1}-\bn). 
\]
Rewriting the second alternative  we obtain  for $\bn\in \cM^l\bigcap\cM^{j+1}$ 
\beq\label{1.12''}
\lefteqn{e^{\im \theta_l} f^l(\bn)\mu^l(\bn)/\nu^l(\bn)}\\
&=&e^{\im \theta} \bar f^{j+1}(\bN+2\bt_{j+1}-\bn +\mbm)\bar \mu^{j+1}(\bN+2\bt_{j+1}-\bn +\mbm)/\nu^{j+1}(\bn),\nn\\
&=&e^{\im \theta} \bar f^{j+1}(\bN+2\bt_{j+1}-\bn +\mbm)\bar \mu^{l}(\bN+2\bt_{l}-\bt-\bn +\mbm)/\nu^l(\bn-\bt),\nn
\eeq
provided that $f^l(\bn )$ and $\bar f^{j+1}(\bN+2\bt_{j+1}-\bn +\mbm)$ are  
both zero or nonzero.

Consider  any $ \bn\in S_0$ (hence  
$f^l(\bn)\neq 0$) and  assume $\bar f^{j+1}(\bN+2\bt_{j+1}-\bn+\mbm)\neq 0$. Otherwise \eqref{1.12''} is false and can be ruled out.

Taking logarithm and rearranging terms in \eqref{1.12''} we have
\beq\label{1.12+}
&&\ln{\mu^l(\bn-\bt)}-\ln{\bar\mu^l(\bN+2\bt_l-\bt-\bn+\mbm)}\\
&=&\im\theta-\im\theta_l-\ln{f^l(\bn)}+\ln{\bar f^{j+1}(\bN+2\bt_{j+1}-\bn+\mbm)}+\ln \alpha(\bn)-\ln \alpha(\bn-\bt)\nn\\
&&+ \im \phi(\bn)-\im\phi(\bn-\bt).\nn
\eeq
The imaginary parts of the lefthand side of \eqref{1.12+} 
\beqn
\Theta_2:=\theta(\bn-\bt)+\theta(\bN+2\bt_l-\bt-\bn+\mbm)
\eeqn
is the sum of two independent random variables unless
\beqn
\bn=\bt_l+\half(\bN+\mbm)
\eeqn
in which case $\Theta_2=2\theta(\bn-\bt)$ is not a sum of two independent random variables. Hence the support of the probability density function
of $\Theta_2$ contains $(-2\gamma\pi,2\gamma\pi]$. 
By the same argument as above, \eqref{1.12+} holds true  with  probability at most $p_2^{|S_0|/2}$ where
\[
p_2:=\max_{a\in \IR}\mbox{\rm Pr} \{\Theta_2\in (a-2\delta\pi, a+2\delta\pi]\}<1 
\]
since $\delta<\min(\gamma,\half)$.

Combining the analysis of the two alternatives,  \eqref{3.100} fails for $k=j+1$
with probability at most $p_1^{|S_0|/2}+p_2^{|S_0|/2}\leq 2p^{|S_0|/2}$ conditioned on the event
that \eqref{3.100} holds true for $k=0,\dots, j$ where $p$ is as  given in \eqref{pgamma}. 
Therefore, the desired result \eqref{3.100} holds with probability
at least $
1-2 Q p^{|S_0|/2}
$ after subtracting the failure probability for each additional block.

\section*{Acknowledgments}
The research of A. F. is supported by  the US National Science Foundation  grant DMS-1413373.  The  present work was initiated during a stimulating and fruitful visit of A.F. to National Center for
Theoretical Sciences (NCTS), Taiwan,  in October 2017.  Research of P. C. is supported in part by the grant 
MOST 107-2115-M-005-006-MY3 from Ministry of Science and Technology, Taiwan. \\

\end{document}